\documentclass[sigconf, screen]{acmart}
\usepackage[utf8]{inputenc}
\renewcommand\footnotetextcopyrightpermission[1]{} % removes footnote with conference information in first column

\title{TeraHAC: Hierarchical Agglomerative Clustering of~Trillion-Edge Graphs}

\usepackage[utf8]{inputenc} % allow utf-8 input
\usepackage[T1]{fontenc}    % use 8-bit T1 fonts

\usepackage{natbib}
\usepackage{graphicx}
\usepackage{fullpage}
\usepackage[capitalise]{cleveref}
\usepackage{amsthm}
\usepackage[frozencache,cachedir=.]{minted}

\usepackage{hyperref}       % hyperlinks
\usepackage{url}            % simple URL typesetting
\usepackage{booktabs}       % professional-quality tables
\usepackage{amsfonts}       % blackboard math symbols
\usepackage{nicefrac}       % compact symbols for 1/2, etc.
\usepackage{microtype}      % microtypography
\PassOptionsToPackage{dvipsnames}{xcolor}
\usepackage{thm-restate}
\usepackage{mathtools}
\usepackage{xspace}
\usepackage{bm}
\usepackage{amsthm}
\usepackage{amsmath}
\usepackage{multirow}
\usepackage{tabularx}
\usepackage{enumitem}
\usepackage{todonotes}

\newcommand{\myparagraph}[1]{\smallskip\noindent {\bf #1.}}

\newtheorem{theorem}{Theorem}
\newtheorem{definition}{Definition}
\newtheorem{lemma}{Lemma}

\crefname{observation}{Observation}{Observations} 

\crefname{claim}{Claim}{Claims}
\newcommand{\defn}[1]{\textbf{\emph{#1}}}
\newtheorem{invariant}{Invariant}
\crefname{invariant}{Invariant}{Invariants}

\newcommand{\minmerge}[0]{\ensuremath{\mathrm{M}}}
\newcommand{\wmax}[0]{\ensuremath{w_{\max}}}
\newcommand{\apxwmax}[0]{\ensuremath{\tilde{w}_{\max}}}

\usepackage{algorithm}
\usepackage[noend]{algpseudocode}
\usepackage{algorithmicx}
\algrenewcommand\algorithmicrequire{\textbf{Input:}}
\algrenewcommand\algorithmicensure{\textbf{Output:}}

\newcommand{\terahac}{$\mathsf{TeraHAC}$}
\newcommand{\subhac}{$\mathsf{SubgraphHAC}$}
\newcommand{\flatten}{$\mathsf{Flatten}$}
\newcommand{\wthreshold}{\ensuremath{t}}
\newcommand{\gness}{\ensuremath{\mathrm{goodness}}}
\newcommand{\scc}{\ensuremath{\mathsf{SCC}}}

\definecolor{forestgreen}{rgb}{0.13, 0.55, 0.13}

\newcommand{\revision}[1]{{{#1}}}
\newcommand{\shepchange}[1]{{#1}}

\newcommand{\id}[1]{\ifmmode\mathit{#1}\else\textit{#1}\fi}
\newcommand{\const}[1]{\ifmmode\mbox{\textc{#1}}\else\textsc{#1}\fi}
\newcommand{\degree}[1]{\ensuremath{deg(#1)}}

\newcommand{\unionnoarg}[1]{\ensuremath{\textsc{Union}}}
\newcommand{\uniontext}[1]{union}

\newcommand{\nghheap}[1]{neighbor-heap}

\newcommand{\STAB}[1]{\begin{tabular}{@{}c@{}}#1\end{tabular}}

\newcommand{\parhac}{$\mathsf{ParHAC}$}
\newcommand{\seqhac}{$\mathsf{SeqHAC}$}
\newcommand{\optrac}{$\mathsf{OptimizedRAC}$}
\newcommand{\rac}{$\mathsf{RAC}$}

\newcommand{\whp}[1]{\emph{whp}}

\definecolor{best}{rgb}{0.0, 0.5, 0.0}
\newcommand{\best}[1]{\color{best}{\underline{#1}}}

\definecolor{munsell}{rgb}{0.0, 0.5, 0.69}
\definecolor{burntsienna}{rgb}{0.91, 0.45, 0.32}

\author{Laxman Dhulipala}
\affiliation{%
 \institution{UMD and Google Research}
 \country{USA} % this nonsense is needed in ACM art...
}
\email{laxman@umd.edu}

\author{Jason Lee}
\affiliation{%
 \institution{Google Research}
 \country{USA} % this nonsense is needed in ACM art...
}
\email{jdlee@google.com}

\author{Jakub Łącki}
\affiliation{%
 \institution{Google Research}
 \country{USA} % this nonsense is needed in ACM art...
}
\email{jlacki@google.com}

\author{Vahab Mirrokni}
\affiliation{%
 \institution{Google Research}
 \country{USA} % this nonsense is needed in ACM art...
}
\email{mirrokni@google.com}

\begin{document}

\begin{abstract}
We introduce \terahac{}, a $(1+\epsilon)$-approximate hierarchical agglomerative clustering (HAC) algorithm which scales to trillion-edge graphs.
Our algorithm is based on a new approach to computing $(1+\epsilon)$-approximate HAC, which is a novel combination of the nearest-neighbor chain algorithm and the notion of $(1+\epsilon)$-approximate HAC.
Our approach allows us to partition the graph among multiple machines and make significant progress in computing the clustering within each partition before any communication with other partitions is needed.

We evaluate \terahac{} on a number of real-world and synthetic graphs of up to 8 trillion edges.
We show that \terahac{} requires over 100x fewer rounds compared to previously known approaches for computing HAC.
It is up to 8.3x faster than SCC, the state-of-the-art distributed algorithm for hierarchical clustering, while achieving 1.16x higher quality.
In fact, \terahac{} essentially retains the quality of the celebrated HAC algorithm while significantly improving the running time.
\end{abstract}

\maketitle
\pagestyle{plain}

\section{Introduction}\label{sec:intro}

Hierarchical agglomerative clustering (HAC) is a widely-used clustering algorithm~\cite{murtagh2012algorithms,murtagh2017algorithms,mullner2011modern, mullner2013fastcluster, gronau2007optimal, stefan1996multiple} known for its high quality in a variety of applications~\cite{zhao2002evaluation, hua2017mgupgma, kobren2017hierarchical, blundell2013bayesian, culotta2007author}.

The algorithm takes as input a collection of $n$ points, as well as a function giving similarities between pairs of points.
At a high level, the algorithm works as follows. Initially, each point is put in a separate cluster of size $1$.
Then, the algorithm runs a sequence of steps.
In each step, the algorithm finds a pair of most similar clusters and merges them together, that is, the two clusters are replaced by their union.
Here, the similarity between the clusters is computed based on the similarities between the points in the clusters.
The exact formula used is configurable and referred to as \emph{linkage function}.

Common linkage functions include \emph{single-linkage} -- the similarity between two clusters $C$ and $D$ is the maximum similarity of two points belonging to $C$ and $D$, and \emph{average-linkage} -- the similarity between $C$ and $D$ is the total similarity between points in $C$ and $D$ divided by $|C|\cdot|D|$.
Due to the particularly good empirical quality, HAC using average linkage similarity is of particular interest~\cite{zhao2002evaluation, hua2017mgupgma, kobren2017hierarchical, moseley-wang, hac-reward, monath2021scalable}.

The output of the algorithm is a \emph{dendrogram} -- a rooted binary tree (or a collection thereof) describing the merges performed by the algorithm.
The leaves of the tree correspond to the initial clusters of size 1, and each internal vertex represents a cluster obtained by merging its two children.

Up until recently, a major shortcoming of HAC was very limited scalability, as the best known algorithms required time that is quadratic in the number of input points.
In addition to that, it was commonly assumed that the input to the algorithm is a complete $n \times n$ matrix giving similarities between all input points, which posed another scalability barrier.\footnote{Even if only a linear number of pairs of points had nonzero similarity, the previously known algorithms would still require quadratic time.}

Recent results overcome the quadratic-time bottleneck by 
allowing $(1+\epsilon)$-approximation~\cite{dhulipala2021hierarchical, parhac}.
A HAC algorithm is $(1+\epsilon)$-approximate if in each step it merges together two clusters whose similarity is within a $(1+\epsilon)$ factor from the maximum similarity between two clusters (at that point).

By allowing approximation and assuming that the input to the algorithm is a sparse similarity graph (e.g. the k-nearest neighbors graph of the input points), it was shown that HAC can be implemented in near-linear time~\cite{dhulipala2021hierarchical} and efficiently parallelized~\cite{parhac}.
At the same time, both relaxations (approximation and using a sparse graph) were shown not to negatively affect the empirical quality of the obtained solutions~\cite{dhulipala2021hierarchical, parhac}.
This line of work led to single machine parallel HAC implementations scaling to graphs containing several billion vertices and about 100 billion edges~\cite{parhac}.

Nevertheless, some large-scale applications require clustering even larger datasets.
To address these needs, several distributed hierarchical clustering algorithms have been proposed~\cite{NIPS2017_2e1b24a6, monath2021scalable, 7184911, jin2013disc, sumengen2021scaling}.
These algorithms can handle up to tens of billions of datapoints and several trillion edges.
However, in order to provide high quality results and/or provable guarantees on the quality of the output, the algorithms need to run typically a significant number of parallel rounds on large-scale datasets~\cite{monath2021scalable, sumengen2021scaling}, and for algorithms that can run in few rounds, no provable quality guarantees are known~\cite{monath2021scalable}.
An intriguing open question, then, is whether it is possible to achieve provable quality guarantees for hierarchical clustering, while running in very few rounds.

\subsection{Our Contribution}
In this paper we introduce \terahac{} -- a $(1+\epsilon)$-approximate HAC algorithm, which runs in very few rounds and is thus amenable to a distributed implementation.
We obtain \terahac{} by developing a new approach to computing $(1+\epsilon)$-approximate HAC, which may be of independent interest.
We evaluate \terahac{} empirically and show that our implementation scales to datasets of tens of billions of points and trillions of edges.

Our key algorithmic contribution is a new approach to computing $(1+\epsilon)$-approximate HAC, which we now outline.
The exact ($1$-approximate) HAC algorithm defines a very specific order, in which cluster merges must be performed -- only a highest similarity pair of clusters can be merged in each step.
However, the very same output dendrogram (up to different tiebreaking) can be computed using a (parallel) RAC algorithm~\cite{sumengen2021scaling} (which is based on the idea behind the nearest-neighbor chain algorithm~\cite{benzecri1982construction, parchain}).
In each step, the RAC algorithm finds all reciprocally most similar clusters.
That is, it finds each pair of clusters $C$ and $D$, such that $D$ is the most similar cluster to $C$, and $C$ is the most similar cluster to $D$ (for simplicity, let us assume that there are no ties).
It is easy to see that the reciprocally most similar pairs of clusters form a matching.
Then, the algorithm merges both clusters of each pair in parallel.
Clearly, this approach allows potentially making multiple merges in one step, some of which may merge pairs of clusters whose similarity is much smaller than the maximum similarity between two clusters.
Crucially, this yields the same result as the basic algorithm, even though both algorithms may perform merges in a very different order.

We generalize the nearest-neighbor chain/RAC algorithm to a $(1+\epsilon)$-approximate algorithm, by introducing a notion of a \emph{good} merge.
Namely, we call the merge of two clusters $(1+\epsilon)$-good if it satisfies a certain condition (see \cref{def:good}).
Crucially, whether a merge of two clusters is good can be checked locally, essentially by looking at the clusters and their incident edges.
We show that by performing the $(1+\epsilon)$-good merges in any order (for example, in parallel) one obtains a $(1+\epsilon)$-approximate dendrogram, i.e., one that can be obtained by performing a sequence of $(1+\epsilon)$-approximate merges.

To put our algorithmic result in context, let us compare the amount of parallelism available to different HAC algorithms by looking at what edges can be merged in the beginning of the algorithm.
Fig.~\ref{fig:merges} compares the amount of parallelism in the RAC algorithm, \terahac{} and ParHAC.
ParHAC is a parallel $(1+\epsilon)$-approximate parallel HAC algorithm~\cite{dhulipala2021hierarchical}, which allows merging any edge whose weight is at most a $(1+\epsilon)$ factor away from the globally maximum weight.
If we consider the sets of edges that can be merged by RAC and ParHAC, it is easy to see that none of them is a superset of the other one.
At the same time, the set of edges that TeraHAC can merge is a superset of both sets (this is the case in general, not just in the provided example).
As shown in Fig.~\ref{fig:terahac_vs_parhac_rac_rounds} the increased parallelism leads to over 100x reduction in the number of rounds.
\shepchange{Importantly, the reduction in the number of rounds also yields between 3.5--10.5x improvements in running time as shown in Fig.~\ref{fig:terahac_vs_optrac_time}.}

\begin{figure}[t]
\begin{center}
\includegraphics[scale=0.35]{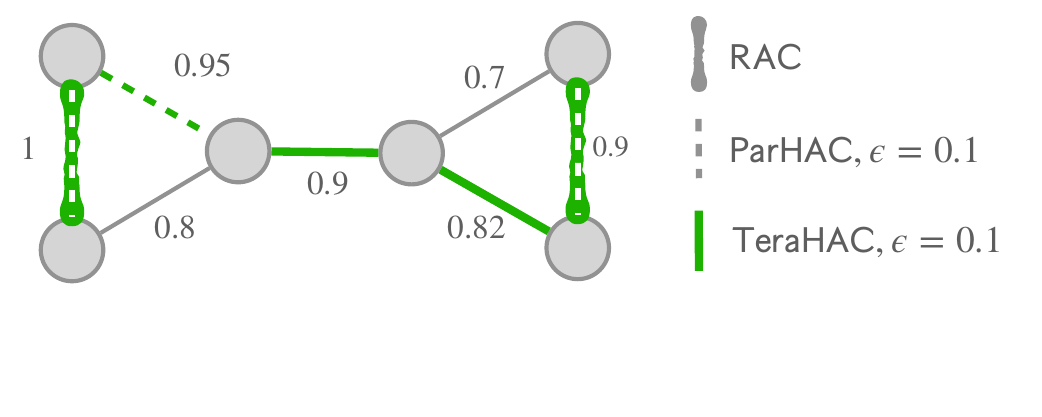}
\vspace{-2em}
\caption{\small Comparison of the merges available at the start of the algorithm for different parallel HAC algorithms.
\label{fig:merges}}
\vspace{-1em}
\end{center}
%\vskip -0.1in
\end{figure}

We evaluate \terahac{} on a number of graph datasets of up to $8$ trillion edges.
We first study its quality on several smaller datasets with ground truth.
We show that when using $\epsilon = 0.1$ \terahac{}, the difference in quality compared to exact HAC is only $1\%$--$3\%$ (according to standard clustering similarity measures, ARI and NMI), which is in line with previous results on $(1+\epsilon)$-approximate HAC~\cite{dhulipala2021hierarchical, parhac}.
Moreover, \terahac{} is more scalable and accurate than \scc{}~\cite{monath2021scalable} -- the state of the art distributed algorithm for hierarchical graph clustering; it is on average 8.3x faster than \scc{} in \scc{}'s highest-quality setting (using 100 rounds)
on a set of large real-world web graphs. For the same high-quality setting, we find that \terahac{} achieves 1.16x higher quality than
\scc{}.

Second, we study the scalability of \terahac{} on a number of large graphs.
We show that for $\epsilon = 0.1$ the number of rounds \terahac{} runs is very low (at most 17 rounds over all our datasets), and that the residual graph shrinks at a geometric rate.
We also find that allowing approximation yields over an order of magnitude more $(1+\epsilon)$-good edges (edges that \terahac{} can potentially merge), compared with edges mergeable by RAC (see Fig.~\ref{fig:clueweb-good-edges}).

Finally, we study the performance and quality of \terahac{} on an 8 trillion-edge graph representing similarities between 31 billion web queries using a set of human-generated labels.
We observe that \terahac{} performs consistently better along a precision-recall curve, in some cases improving recall by up to $20\%$ (relative), while being $30\%$--$50\%$ faster.

\subsection{Further Related Work}
\revision{
The HAC algorithm is particularly useful for semantic clustering applications, where the desired number of clusters is not known upfront, for example for de-duplication or, more generally, detecting groups of closely related entities~\cite{zhao2002evaluation, hua2017mgupgma, monath2021scalable, kobren2017hierarchical}.
In these cases in large-scale applications the number of clusters of interest is very high (close to the number of input entities), and so the average cluster size is low~\cite{kobren2017hierarchical}.

For other prominent applications of clustering, multiple successful algorithms have been developed.
For example DBSCAN~\cite{schubert2017dbscan} and correlation clustering~\cite{bansal2004correlation} are particularly well suited for finding anomalously dense clusters.
Modularity clustering~\cite{newman2006modularity} is widely used for community detection.
Graph partitioning methods~\cite{partitioning2013charles, bulucc2016recent} split the data into a typically small number of roughly-equally sized clusters, while ensuring that related entities are not split across different clusters (hence, it is in a sense a dual problem to the one addressed by HAC).
Finally, k-means/k-median/k-center methods as well as spectral clustering~\cite{ng2001spectral, dhillon2004kernel} are particularly well-suited when the desired number of clusters is known upfront.
}

The notion of approximate HAC~\cite{48657} was used to give more efficient HAC algorithms in the case when the input is a metric space~\cite{48657, pmlr-v130-moseley21a, DBLP:conf/icml/YaroslavtsevV18, NEURIPS2019_d98c1545} as well as a graph~\cite{dhulipala2021hierarchical, parhac}.

For average linkage HAC, to the best of our knowledge, the fastest single-machine implementations scale to tens of millions of points in metric space~\cite{parchain} or billion-vertex graphs~\cite{parhac}.

A number of graph clustering problems have been studied in the distributed setting, for example in MapReduce~\cite{62}, Pregel/Giraph~\cite{malewicz2010pregel, sakr2016large} and other frameworks captured by the MPC model of computation~\cite{karloff2010model}.
Examples include efficient algorithms for correlation clustering~\cite{10.1145/2623330.2623743, pmlr-v139-cohen-addad21b}, connected components~\cite{lkacki2018connected} and balanced partitioning~\cite{10.14778/3324301.3324307, 10.1145/2835776.2835829}.

For HAC, efficient distributed algorithms are known for the single-linkage variant~\cite{7184911, jin2013disc, DBLP:conf/icml/YaroslavtsevV18}.
We note that single-linkage hierarchical agglomerative clustering generally delivers worse quality compared to average linkage~\cite{zhao2002evaluation, hac-reward}.
\shepchange{
Several HAC algorithms have been recently developed in the shared-memory setting, including \seqhac{}, a nearly linear time $(1+\epsilon)$-approximate sequential HAC algorithm~\cite{dhulipala2021hierarchical} and \parhac{}, a nearly linear work $(1+\epsilon)$-approximate HAC algorithm with low depth~\cite{parhac}.
\parhac{} provably runs in poly-logarithmic rounds~\cite{parhac}, whereas \terahac{} uses a novel relaxation of the reciprocal clustering approach of RAC, and neither \terahac{} nor RAC are known to admit non-trivial round-complexity bounds.
However, empirically \terahac{} uses the fewest rounds of any modern HAC algorithm we study (see Figure~\ref{fig:threshold_vs_rounds}).
}

\begin{figure}[t]
\begin{center}
\vspace{-0.25em}
\includegraphics[width=0.3\textwidth]{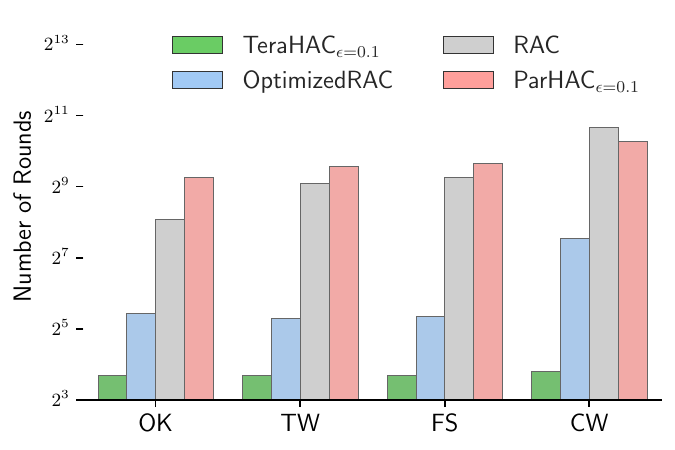}
\vspace{-1em}
\caption{\small Number of rounds used by \terahac{} compared with 
\optrac{} (\terahac{} using $\epsilon=0$),  
\parhac{} 
and 
 \rac{} on four large real-world graph datasets. All algorithms use a weight
 threshold of $t=0.01$ (see Section~\ref{sec:experiments}).
\label{fig:terahac_vs_parhac_rac_rounds}}

\vspace{-1.5em}
\end{center}
\end{figure}

\begin{figure}[t]
\begin{center}
\vspace{-0.25em}
\includegraphics[width=0.3\textwidth]{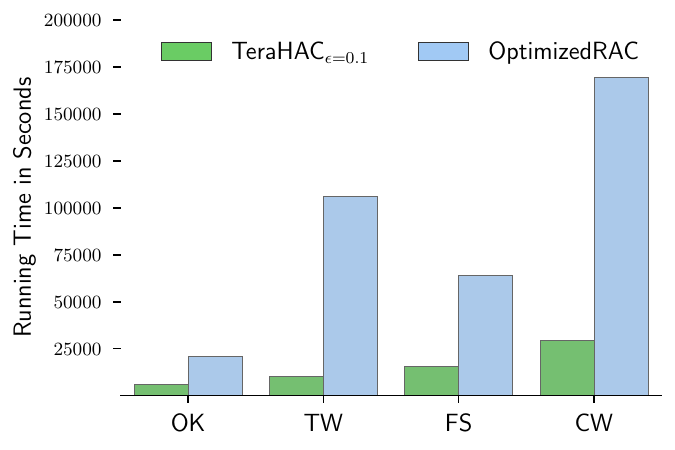}
\vspace{-1em}
\caption{\small 
\shepchange{
Distributed running times of \terahac{} compared with 
\optrac{} (\terahac{} using $\epsilon=0$) on the same graphs and threshold
as Figure~\ref{fig:terahac_vs_parhac_rac_rounds}.}
\label{fig:terahac_vs_optrac_time}}

\vspace{-1em}
\end{center}
%\vskip -0.1in
\end{figure}

The most scalable hierarchical clustering algorithms have been obtained by slightly diverging from the HAC algorithm definition.
Affinity clustering~\cite{NIPS2017_2e1b24a6}, and its successor SCC~\cite{monath2021scalable} are two highly scalable hierarchical clustering algorithms designed for the MapReduce model.
Both take graphs as input and were shown to scale to graphs with trillions of edges.
While both algorithms are similar to HAC by design, in practice they can give very different results from what the HAC algorithm gives.
We compare with SCC in our empirical analysis.
We conclude that while SCC can be tuned to obtain quality which is comparable with HAC, this comes at a cost of significantly higher running time of SCC.
We note that SCC was shown to deliver higher quality than affinity clustering.

Another highly scalable HAC implementation is the RAC algorithm~\cite{sumengen2021scaling} obtained by parallelizing the nearest neighbor chain algorithm.
The algorithm was shown to scale to datasets of billion points.
Since RAC is based on the exact HAC algorithm, the number of rounds it uses is an order of magnitude larger than what can be achieved by $(1+\epsilon)$-approximate HAC~\cite{parhac}, which negatively affects performance.
\shepchange{As we show, \terahac{} using $\epsilon = 0$ can be viewed as an optimized implementation of RAC, which we refer to as \optrac{}.
At the same time \terahac{} using $\epsilon=0$ uses up to two orders of magnitude fewer rounds (see Figure~\ref{fig:terahac_vs_parhac_rac_rounds}).
}

\section{Preliminaries}

Let $G=(V, E, w)$ be a weighted and undirected graph, where $|V| = n$.
We assume that all edge weights of $G$ are positive.

Let us now describe how the HAC algorithm works given $G$ as an input.
Initially there are $n$ clusters, each containing a distinct vertex of $G$.
The algorithm proceeds in at most $n-1$ steps.
In each step, the algorithm finds two clusters of nonzero similarity and merges them together.
Let $C_1, \ldots, C_k$, be a sequence of clusterings, where $C_i$ is the clustering obtained after running $i-1$ steps of the algorithm.
In particular, $C_1$ is the clustering containing $n$ clusters of size $1$.

We define a sequence of graphs $G_1, \ldots, G_k$, where $G_i$ is obtained from $G$ by contracting the clusters $C_i$.
Specifically, the vertex set of $G_i$ is $C_i$, which means that each vertex of $G_i$ is a \emph{set} of vertices of $G$.
We define the function $w : 2^V \times 2^V  \rightarrow \mathbb{R}$ giving edge weights in $G_i$ as $w(u, \shepchange{v}) = \sum_{xy \in ((u \times v) \cap E)} w(xy) / (|u| \cdot |v|)$. Note that the weight function $w$ is the same for all graphs $G_i$. There is an edge between two vertices $u$ and $v$ in $G_i$ if and only if $w(u,v) > 0$.

\begin{definition}[\cite{benzecri1982construction}]\label{def:reducibility}
A cluster similarity function $f : 2^V \times 2^V  \rightarrow \mathbb{R}$ is called \emph{reducible} if and only if for any clusters $x, y, z$, we have $f(x, y \cup z) \leq \max(f(x,y), f(x,z))$.
\end{definition}

In particular, the function $w$ defined above is reducible~\cite{benzecri1982construction}.
While in the following we work with this specific similarity function, we note that our main theoretical results hold for any similarity function satisfying Definition~\ref{def:reducibility}, which includes 
e.g. single linkage, complete linkage and median linkage.

Note that $G_{i+1}$ is obtained from $G_i$ by contracting a single edge and updating the weights of the edges incident to the vertex created by the contraction.
We call each such operation a \emph{merge} of the endpoints of the contracted edge.
Let $w_{max}(G_i)$ be the largest edge weight in $G_i$.
We say that a merge of an edge $uv$ in $G_i$ is $(1+\epsilon)$-approximate, if $(1+\epsilon) w(uv) \geq w_{max}(G_i)$.

A \emph{dendrogram} is a tree describing all merges performed by running a HAC algorithm.
The tree has between $n$ and $2n-1$ nodes\footnote{We use \emph{nodes} when talking about trees and \emph{vertices} in the context of graphs.}
corresponding to clusters that are created in the course of the algorithm.
Specifically, there are $n$ leaf nodes corresponding to the initial singleton clusters.
In addition, for each merge that creates a cluster $x \cup y$ by merging $x$ and $y$, there is a node $x \cup y$, whose children are $x$ and $y$.
Hence, each internal node of the dendrogram has exactly two children.

In the description of the algorithm, instead of using a dendrogram, it is more convenient to use a \emph{merge tree}.
Observe that if we remove all leaf nodes from the dendrogram, there is a natural 1-to-1 correspondence between the remaining nodes and the merges made by the algorithm.
We call the tree obtained this way, in which each node is a \emph{merge}, a merge tree.

Observe that there may be multiple different \emph{sequences} of merges made by the algorithm that produce the very same merge tree (and dendrogram).
Specifically, given a fixed merge tree, all we know is that the first merge made by the algorithm must have been some merge corresponding to one of its leaves.

Given a merge tree, and a sequence $m_1, \ldots, m_k$ of $k$ distinct merges, we say that the sequence is \emph{consistent} with the tree, if and only if for each $1 \leq i \leq k$, either $m_i$ is a leaf in the merge tree, or $m_i$ is an internal node of the merge tree and the children of $m_i$ are $m_a$ and $m_b$, where $\max(a,b) < i$.

We say that a dendrogram is \emph{$(1+\epsilon)$-approximate} if and only if there exists a sequence of $(1+\epsilon)$-approximate merges, such that (a) the sequence is consistent with the merge tree of the dendrogram, and (b) after performing all merges in the sequence, no pair of clusters has nonzero similarity.
Finally, a $(1+\epsilon)$-approximate HAC algorithm is one which produces a $(1+\epsilon)$-approximate dendrogram.
\section{Approximate Nearest-Neighbor Chain Algorithm}
In this section we extend the 40-year old nearest-neighbor chain algorithm~\cite{benzecri1982construction} (NN-chain) with the notion of approximation.
Let $G$ be an edge-weighted graph.
For each vertex $v$ of $G$, we denote by $\wmax(v)$ the maximum weight of any edge incident to $v$.
It follows from the NN-chain / RAC algorithm that the following algorithm is equivalent to the (exact) HAC algorithm.
As long as the graph has nonzero edges, pick any edge $uv$, such that $w(uv) = \wmax(u) = \wmax(v)$ and merge together its endpoints.

Somewhat surprisingly, even though the HAC algorithm specifies a concrete order of merging edges, the NN-chain algorithm computes the very same dendrogram (up to ties among edge weights) as the standard HAC algorithm.
At the same time, the NN-chain algorithm has an important feature: the decision on whether to merge two vertices can be made entirely locally, which is a very useful property in a parallel setting.
In this section, we give an $(1+\epsilon)$-approximate HAC algorithm also based on a simple local criterion for deciding whether two vertices can be merged, which we define below.

\begin{definition}[Good merge]\label{def:good}
Let $\epsilon \geq 0$ and $G_i$ be a graph obtained from $G$ by performing some sequence of merges.
For each vertex $v$ of $G_i$, we define $\minmerge(v)$ to be the smallest linkage similarity among all merges that were performed to create cluster $v$.
Specifically, for each singleton cluster $v$ we have $\minmerge(v) = \infty$, and whenever two clusters $u$ and $v$ are merged and create cluster $u \cup v$, we have $\minmerge(u \cup v) = \min(\minmerge(u), \minmerge(v), w(uv))$.
With this notation, we say that a merge of an edge $uv$ in $G_i$ is $(1+\epsilon)-$\emph{good} if and only if
\[
\frac{\max(\wmax(u), \wmax(v))}{\min(\minmerge(u), \minmerge(v), w(uv))} \leq 1+\epsilon.
\]
\end{definition}

When $\epsilon$ is clear from the context or irrelevant, we will sometimes say good instead of $(1+\epsilon)$-good.
Observe that the value of $\minmerge(u)$ only depends on the sequence of merges which created $u$.
In contrast, $\wmax(u)$ is a function of both $u$ and some merges outside of $u$, as it also depends on the current set of neighbors of $u$.
This means that $\minmerge(u)$ is well defined given a particular merge tree.
However, to compute $\wmax(u)$ we also need to specify which merges have been performed.

\begin{figure}
    \centering
\vspace{-1em}
    \includegraphics[width=\columnwidth]{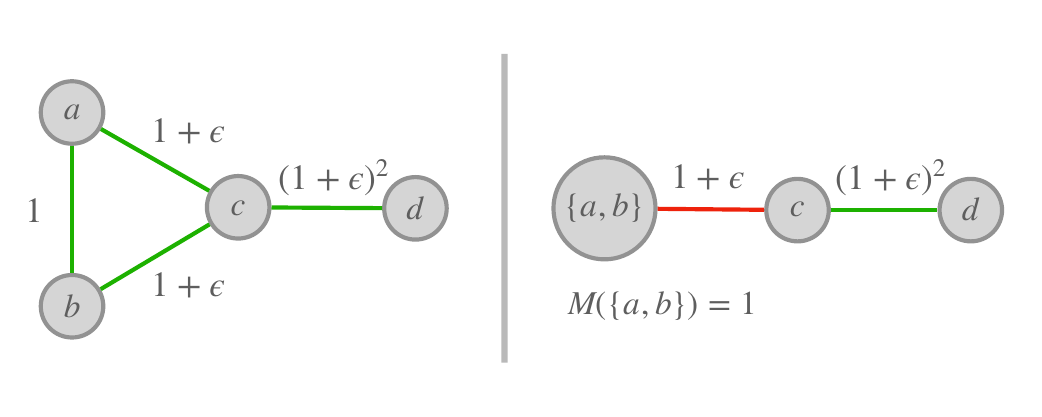}
%   Graph with 4 nodes and 4 edges. $ab$ of weight 1, $bc$ and $ac$ of weight $1+\epsilon$ and $cd$ of weight $(1+\epsilon)^2$. If we have space, we can also add a graph after merging $ab$, in which the merged node is called $\{a, b\}$
\vspace{-3em}
    \caption{\small Example showing the need for $M(\cdot)$ values in Definition~\ref{def:good}. 
    Green edges correspond to merges which are  $(1+\epsilon)$-good, and red edges correspond to merges which are not $(1+\epsilon)$-good.
    After merging $ab$ (which is $(1+\epsilon)$-good) we obtain a vertex $\{a, b\}$ such that $M(\{a, b\}) = 1$. 
    Therefore, the merge of $\{a, b\}$ with $c$ in the resulting graph is \emph{not} $(1+\epsilon)$-good, since $\max(1+\epsilon, (1+\epsilon)^2) / \min(1, \infty, 1+\epsilon) = (1+\epsilon)^2 > 1+\epsilon$. 
    Hence, the algorithm is forced to merge $c$ with $d$. It is easy to see that allowing a merge of $\{a, b\}$ with $c$ would create a dendrogram, which is not $(1+\epsilon)$-approximate.}
    \label{fig:minmerge}
\end{figure}

We will show that applying $(1+\epsilon)$-good merges leads to a $(1+\epsilon)$-approximate dendrogram (see Lemma~\ref{lem:apxfull}).
For $\epsilon = 0$ this generalizes the RAC algorithm, as in this case the $M(u)$ and $M(v)$ terms in the denominator become irrelevant (see Observation~\ref{obs:eqrac}).
However, they are crucially important for $\epsilon > 0$, as shown in Figure~\ref{fig:minmerge}.

In the initial graph $G_1$, a merge is good iff we have that $\max(\wmax(u), \wmax(v)) / w(uv) \leq 1+\epsilon$, regardless of the maximum weight of an edge in the graph.
In general, each $(1+\epsilon)$-approximate is $(1+\epsilon)$-good, but the converse is not necessarily true.
Nevertheless, we can show that performing a sequence of good merges produces a 
$(1+\epsilon)$-approximate dendrogram.
Before we formally state and prove this property, we first show a few auxiliary lemmas.

\begin{restatable}{lemma}{wmaxdecreases}\label{lem:wmaxdecreases}
Let $G_1, \ldots, G_n$ be a sequence of graphs, in which each graph is obtained from the previous one by performing an arbitrary merge.
Let $v$ be a vertex (cluster) which exists in $G_l, \ldots, G_r$.
Let $\wmax^i(v)$ be the value of $\wmax(v)$ in $G_i$, where $l \leq i \leq r$.
Then $\wmax^l(v) \geq \wmax^{l+1}(v) \geq \ldots \geq \wmax^r(v)$.
\end{restatable}

Now, we prove a useful invariant, which is satisfied when performing good merges.
%\laxman{Do we use \cref{lem:goodinvariant} other than the observation for RAC?}

\begin{lemma}\label{lem:goodinvariant}
Let $G_i$ be a graph obtained from $G_1$ by performing $(1+\epsilon)$-good merges.
Then, for each vertex $v$ of $G_i$, we have that $\wmax(v)/\minmerge(v) \leq 1+\epsilon.$
\end{lemma}

\begin{proof}
We prove the claim by induction on the number of good merges made.
The base case (no merges) follows trivially from the fact that $\minmerge(v) = \infty$ for each vertex $v$.
Now, consider that some number of good merges have been made, and the resulting graph is $G_{i+1}$.
Fix a vertex $v$.
If the vertex $v$ also exists in $G_i$, the property follows by \cref{lem:wmaxdecreases}.
Otherwise, $v$ is created by merging two vertices $v_1$ and $v_2$ ($v = v_1 \cup v_2$) using a $(1+\epsilon)$-good merge.
We have
\[
\frac{\wmax(v)}{\minmerge(v)} \leq \frac{\max(\wmax(v_1), \wmax(v_2))}{\minmerge(v)} \leq 1+\epsilon.
\]
In the first inequality we used reducibility (\cref{def:reducibility}), which implies $\wmax(v ) \leq \max(\wmax(v_1), \allowbreak \wmax(v_2))$.
The second inequality follows from the definition of a good merge.
\end{proof}

\begin{restatable}{observation}{eqrac}\label{obs:eqrac}
Let $G_i$ be obtained from $G_1$ by performing $1$-good merges.
Then, a merge $uv$ in $G_i$ is $1$-good if and only if $w(uv) = \max(\wmax(u), \wmax(v))$.
\end{restatable}

\noindent The proof uses Lemma~\ref{lem:goodinvariant};
we provide it in the Appendix.

\begin{definition}
Let $D$ be a dendrogram and $M$ be its corresponding merge tree.
Consider a sequence of merges $m_1, \ldots, m_k$ which contains all merges of $M$ and is consistent with $M$.
We say that this sequence is a \emph{greedy merge sequence} of $D$ if its obtained as follows.
Start with an empty sequence.
At each step append a merge $m_i$ of maximum weight, chosen such that the new expanded sequence is consistent with $M$.
\end{definition}

A greedy merge sequence is a {\emph canonical} merge sequence, in the sense that it achieves the best approximation ratio (i.e. the maximum approximation ratio of a merge is minimized) out of all consistent orderings of merges of $M$.

Let $D$ be a dendrogram and $M$ be its corresponding merge tree.
Let $m_1, \ldots, m_k$ be a sequence of merges consistent with $M$ and $G_k$ be the graph obtained after applying them.
We say that the \emph{error} of a merge $m$ which merges $u$ and $v$ is the maximum weight of an edge in $G_k$ divided by $w(uv)$.

\begin{lemma}\label{lem:empapx}
Let $G=(V, E,w)$ be a weighted graph, $D$ be a dendrogram of $G$ and $m_1, \ldots, m_n$ be its greedy merge sequence. 
Let $1+\epsilon$ be the maximum error of a merge in the produced sequence.
Then, $D$ is $(1+\epsilon)$-approximate and is not $(1+\epsilon')$-approximate for any $\epsilon' < \epsilon$.
\end{lemma}

\begin{proof}
Let $q_1, \ldots, q_n$ be an optimal merge sequence consistent with $M$, that is, a sequence in which the maximum error of a merge is minimum possible.
Let $1+\epsilon$ be the maximum error of a merge in this sequence.

We refer to the algorithm from the lemma statement as the greedy algorithm.
Note that the merge sequence $m_1, \ldots, m_n$  it produces is a permutation of $q_1, \ldots, q_n$, as both sequences consist of all merges of $M$.
It suffices to prove that each merge appended by the algorithm has error at most $1+\epsilon$.
We prove that this property holds for each prefix $m_1, \ldots, m_k$ by induction on $k$.
The case of $k=0$ is trivial. Consider $k > 0$.

Let $j \geq 0$ be the minimum value such that $q_{j+1}$ is \emph{not} one of $m_1, \ldots, m_k$.
This implies that the set of merges $\{m_1, \ldots, m_k\}$ is a superset of $\{q_1, \ldots, q_j\}$.
This in turn implies that the maximum weight in the graph after applying merges $m_1, \ldots, m_k$ is not larger than the maximum weight in the graph after applying $q_1, \ldots, q_j$.
Hence, the error of the merge $q_{j+1}$ after applying $m_1, \ldots, m_k$ is at most the error of $q_{j+1}$ after applying $q_1, \ldots, q_j$, which by definition is at most $1+\epsilon$.
Since the greedy algorithm chooses the merge of minimum error, we conclude that the error of $m_{k+1}$ is at most $1+\epsilon$.
\end{proof}

We conclude this section with our main theoretical result, which shows that applying $(1+\epsilon)$-good merges produces a $(1+\epsilon)$-approximate dendrogram.
\revision{We stress that the $(1+\epsilon)$-good merges can be applied in \emph{any} order (as long as they are $(1+\epsilon)$-good in that order).
This allows us to remove the greedy aspect of HAC from consideration and achieve better parallelism.} 

\begin{lemma}\label{lem:apxfull}
Let $D$ be a dendrogram of $G$ obtained by applying $(1+\epsilon)$-good merges and $\epsilon > 0$.
Then, \revision{$D$ is} $(1+\epsilon)$-approximate.
\end{lemma}

\begin{proof}
Let $m_1, \ldots, m_r$ be a greedy merge sequence of $D$.
We use induction on $k$ to show that all merges $m_1, \ldots, m_k$ are $(1+\epsilon)$-approximate.
For $k=0$ the lemma is trivial, so let us assume $0 < k \leq r$.
Let $G_k$ be the graph obtained by applying the merges $m_1, \ldots, m_{k-1}$.
Let $xy$ be the maximum weight edge in $G_k$ (chosen arbitrarily in case of a tie).

We will say that a merge $m'$ is \emph{available} if it's one of the merges in the merge tree of $D$, and $m_1, \ldots, m_k, m'$ is consistent with the merge tree of $D$.
Our goal is to show that $m_k$ is $(1+\epsilon)$-approximate.
Hence, we need to show that the merge similarity of $m_k$ is at least $w(xy) / (1+\epsilon)$.
%\laxman{iiuc, we need to show that the edge has error $(1+\epsilon)$; shouldn't we show that $m_k \geq \mathsf{maxweight}(G_k) / (1+\epsilon)$?}
To that end, it suffices to show that there is an available merge of similarity at least $w(xy) / (1+\epsilon)$.
This is because, thanks to the greedy construction, $m_k$ is the available merge of maximum merge similarity.

Consider the first merges in the dendrogram $D$, which involved vertices $x$ and $y$.
That is, let $x'$ be the sibling vertex of $x$ in the dendrogram $D$ and $y'$ be the sibling vertex of $y$ (i.e., $x'$ is the cluster with which $x$ is merged in the dendrogram).

If $x' = y$, then also $y' = x$ and the merge of $x$ with $y$ is available.
The weight of this merge is $w(xy)$, which is trivially more than $w(xy) / (1+\epsilon)$.
This concludes the proof in this case.

Now consider the case when $x' \neq y$ (which also means that $y' \neq x$).
Without loss of generality, assume that in the merge sequence which produced dendrogram $D$, the merge of $x'$ with $x$ happened before the merge of $y'$ with $y$.

\begin{restatable}{claim}{clmerge}\label{claim:minmerge}
We have $w(xy) / \min(\minmerge(x'), \minmerge(x), w(xx')) = \\ w(xy) / \minmerge(x' \cup x) \leq 1+\epsilon$.
\end{restatable}

\revision{
The proof can be found in the Appendix.}

Consider the subtree of the merge tree rooted at $x' \cup x$.
We will show that that all merges in that subtree have merge similarity at least $w(xy) / (1+\epsilon)$.
This would finish the proof, as the subtree clearly contains at least one available merge.

Consider any available merge in that subtree and assume that it merges together $z_1$ and $z_2$, where $z_1, z_2$ are both subsets of $x' \cup x$.
By the definition of $\minmerge$ and \cref{claim:minmerge}, $w(z_1z_2) \geq \minmerge(x' \cup x) \geq w(xy) / (1+\epsilon)$, which completes the proof.
\end{proof}

\begin{algorithm}[!t]\caption{\terahac{}($G=(V, E, w, \minmerge), \epsilon, \wthreshold$)}\label{alg:terahac}
    \begin{algorithmic}[1]
    \Require{Similarity graph $G$, $\epsilon > 0$, $\wthreshold \geq 0$.}
    \Ensure{$(1+\epsilon)$-approximate HAC dendrogram.}
    \While{$|E| > 0$}\label{parhac:loopstart}
      \State Partition $G$ into clusters $C'_1, \ldots, C'_k$ \label{l:partition}
      \For{each cluster $C$ in $C'_i$ in parallel}
         \State $G^C := $ subgraph of $G$ consisting of vertices in $C$ \indent \indent and all their incident edges
         \State $M := $ \subhac{$(G^C, C)$}
         \State Apply the merges $M$ in the graph $G$
      \EndFor
      \State Remove from $G$ vertices $v$ s.t. $\wmax(v) < \wthreshold / (1+\epsilon)$\label{l:optvertices}
      \EndWhile
    \end{algorithmic}
\end{algorithm}

\begin{algorithm}[!t]\caption{\subhac{}($G=(V, E, w, \minmerge), A, \epsilon$)}\label{alg:subhac}
    \begin{algorithmic}[1]
    \Require{$G^A = (V, E, w, \minmerge)$, $A \subseteq V$, $\epsilon > 0$.}
    \Ensure{A set of merges of vertices in $A$}
    \State Mark vertices of $V \setminus A$ as \emph{inactive}
    \State $T := $ priority queue with edges of $G[A]$ keyed by goodness.
    \While{$\min(T) \leq 1+\epsilon$}
        \State $uv :=  \textsc{RemoveMin}(T)$
        \State Remove from $T$ all edges incident to $u$ or $v$
        \State Merge $u$ and $v$ in $G$, and update the weights of all edges \indent incident to $u \cup v$ in $G$
        \State Add to $T$ incident edges of $u \cup v$ that do not have an \indent inactive endpoint
    \EndWhile
    \State Return all merges made
    \end{algorithmic}
\end{algorithm}

\section{\terahac{} algorithm}

The \terahac{} algorithm is presented as Algorithm~\ref{alg:terahac}. It takes as input a weighted graph $G$, the accuracy parameter $\epsilon$ and a weight threshold $\wthreshold > 0$.
Each vertex $v$ of $G$ is assigned a $\minmerge(v)$ value which is initially equal to $\infty$.
The weight threshold $\wthreshold$ is used to prune the dendrogram by performing {\em vertex pruning}, i.e., stop the algorithm once it performs all merges of sufficiently high similarity.
The effect of this parameter will be discussed later. For now, let us assume that $t = 0$, which makes line \ref{l:optvertices} redundant.
We analyze the $t > 0$ case in Section~\ref{sec:flattening}.

The algorithm runs in a loop. Each iteration of the loop is a {\em round}, which merges together some vertices of $G$.
The loop runs until the graph has no edges.
In each round, we first partition the graph $G$ into clusters.
The choice of the partitioning method affects the running time of the algorithm, but from the correctness point of view, the partitioning can be arbitrary.

Given a graph $G=(V, E, w)$ and a subset of its vertices $C$, we use $G^C$ to denote a graph consisting of all vertices of $C$ and their immediate neighbors, as well as edges that have at least one endpoint in $C$.
Formally, the vertex set of $G^C$ is \shepchange{$C \cup \{x \in V \mid \exists_{y \in C} xy \in E\}$} and the edge set is $\{xy \in E \mid x \in C \vee y \in C\}$.

Once the partitioning is computed, for each cluster $C$ we compute a graph .
Observe that each intra-cluster edge belongs to exactly one graph and each inter-cluster edge belongs to exactly two graphs computed this way.

Next, the algorithm runs \subhac{} algorithm on each graph $G^C$.
The goal of \subhac{} is to find a longest sequence of merges which are $(1+\epsilon)$-\emph{good} (see Definition~\ref{def:good}).
The challenge is that while the algorithm only sees a subgraph $G^C$ of $G$, the merges it makes should be good with respect to the entire graph $G$.
Once we have computed all merges made by all \subhac{} calls, we apply them in the global graph.

Let us now describe \subhac{} algorithm in detail.
The algorithm is given a set $C$ of vertices of $G$ and a graph $G^C$.
The vertices of $C$ in $G^C$ are called \emph{active}, and the remaining vertices are \emph{inactive}.
The goal of the algorithm is to perform a maximal sequence of good merges of active vertices in $C$.
We assume that vertices formed by merging active vertices are also active.
We later show that the good merges in $G^C$ are also good merges in the "global" graph $G$.

To decide whether a merge of $u$ and $v$ is good, we need to know all incident edges of $u$ and $v$.
This is why the graph $G^C$ contains not only vertices of $C$, but also their neighbors.
Since each merge can only involve vertices of $C$ (i.e., they do not affect the inactive vertices), the merges performed by parallel invocations of \subhac{} affect disjoint sets of vertices.

Given an edge $uv$ of $G^C$, where both $u$ and $v$ are active (that is, not inactive) we define $\gness(uv) = \max(\wmax(u),\allowbreak \wmax(v)) / \minmerge(u \cup v)$.
By Definition~\ref{def:good}, a merge of $u$ with $v$ is $(1+\epsilon)$-good iff $\gness(uv) \leq 1+\epsilon$.\footnote{Somewhat counter-intuitively edges of  \emph{lower} goodness are better. Renaming \emph{goodness} to \emph{badness} would be more intuitive, but would also sound negative.}
At a high level, in each step the algorithm finds the merge of the smallest goodness.
If the goodness is more than $1+\epsilon$, it terminates.
Otherwise, it performs the merge and continues.

We give a near-linear time implementation of \subhac{} which performs a nearly-maximal
set of $(1+\epsilon)$-good merges.
More specifically, the algorithm is guaranteed to perform {\em all merges} in the subgraph with goodness in the range $[1, T)$ where $T = 1 + \Theta(\epsilon)$, and no merges with goodness $> 1+\epsilon$.
\revision{
We bound the running time of \subhac{} as follows. 
\begin{theorem}
The running time of \subhac{} \shepchange{on a graph containing $n$ vertices and $m$ edges} is $O((m + n) \log^2 n)$.
\end{theorem}
We provide the proof of this result in the Appendix.}
Our algorithm is similar to a recent $(1+\epsilon)$-approximate HAC algorithm~\cite{dhulipala2021hierarchical}, but is significantly more complicated due to needing to maintain not only edge weights as merges are performed, but also the  goodness values of edges.

Let us now discuss the correctness of the algorithm.
It is easy to see that each merge \subhac{} performs is good (with respect to its input graph), and that it returns once no good merges can be made.
However, it is not obvious whether the merges computed by all \subhac{} calls are also good when applied on the global graph.
In particular, there are two possible issues.
First, each \subhac{} calls only sees a subgraph of the entire graph and computes good merges with respect to that subgraph.
Second, all merges produced by the parallel calls are then applied on the entire graph, but we do not specify in what order they need to be applied.
We address these in the following two lemmas.

\begin{restatable}{lemma}{lemsub}
For some $C \subseteq V$, consider the graph $G^C$ and let $m_1, \ldots, m_k$ be a sequence of merges in $G^C$, such that (a) the merges are $(1+\epsilon)$-good with respect to $G^C$ and (b) the merges only involve merging active vertices of $G^C$.
Then, the merges $m_1, \ldots, m_k$ are also $(1+\epsilon)$-good with respect to $G_i$.
\end{restatable}

\begin{lemma}
Let $\epsilon \geq 0$, and $G=(V,E,w)$ be a graph obtained by performing a sequence of $(1+\epsilon)$-good merges.
Consider $C \subseteq V$.
Let $m_1, \ldots, m_k$ be a sequence of $(1+\epsilon)$-good merges with respect to $G$, such that each merge involves two vertices of $C$ (or vertices created by merging them).
Moreover, let Let $m'_1, \ldots, m'_l$ be a sequence of $(1+\epsilon)$-good merges with respect to $G$, each involving two vertices of $V \setminus C$ (or vertices created by merging them).

Then, any interleave of these sequences, that is any sequence of length $k+l$ which contains both sequences as subsequences, is a sequence of $(1+\epsilon)$-good merges with respect to $G$
\end{lemma}

\begin{proof}
Whether a merge of $u$ and $v$ is good depends only on $w(uv)$, $\minmerge(u)$, $\minmerge(v)$, $\wmax(u)$, $\wmax(v)$. Among these, only $\wmax(u)$ and $\wmax(v)$ can change due to merges not involving vertices of $D$.
However, by Lemma~\ref{lem:wmaxdecreases}, these values only decrease as a result of merges.
By Definition~\ref{def:good}, the lemma follows.
\end{proof}

This lets us show the correctness of the \terahac{} algorithm for any partitioning method for the case when $t = 0$.

\begin{restatable}{lemma}{lemint}\label{lem:interleave}
Let $D$ be a dendrogram computed by Algorithm \ref{alg:terahac} for $t = 0$ and any partitioning method used in line~\ref{l:partition}.
Then, $D$ is a $(1+\epsilon)$-approximate dendrogram.
\end{restatable}

\revision{
\begin{proof}
By Lemma~\ref{lem:apxfull} it suffices to show that $D$ is obtained by performing a sequence of $(1+\epsilon)$-good merges.
Consider a single iteration of Algorithm~\ref{alg:terahac}.
Denote by $G$ the current graph in the beginning of the iteration.
The iteration begins by computing a clustering $C_1, \ldots, C_k$ of $G$.
Then, it runs \subhac{} separately in each cluster.
Hence, for each cluster $C_i$ we obtain a sequence of $(1+\epsilon)$-good merges with respect to $G$, which we denote by $M_i$.
The resulting dendrogram can be updated with each of these sequences of merges in parallel, since each of these sequences affects a disjoint set of vertices of $G$.

We now use Lemma~\ref{lem:interleave} to show that any interleave of the sequences $M_1, \ldots, M_k$ is a sequence of $(1+\epsilon)$-good merges.
This can be done by using a simple induction.
For $k=1$ the claim follows immediately, as the only sequence that can be obtained by interleaving sequences from the set $\{M_1\}$ is $M_1$ itself.
Now assume that for $i \geq 1$ we know that any interleave of sequences $M_1, \ldots, M_i$ is consists of $(1+\epsilon)$-good merges.
Consider any interleave $\tilde{M}$ of $M_1, \ldots, M_{i+1}$.
Clearly $\tilde{M}$ can be obtained by interleaving some interleave of $\tilde{M}_i$ of $M_1, \ldots, M_i$  with $M_{i+1}$.
We now use Lemma~\ref{lem:interleave} applied to sequences $\tilde{M}_i$ and $M_{i+1}$.
The former consists of $(1+\epsilon)$-good merges by the inductive assumption and the second one by the correctness of \subhac{}.
To apply the lemma we set $C = \bigcup_{j=1}^i C_i$.
Hence, we get that $\tilde{M}$ consists of $(1+\epsilon)$-good merges.
\end{proof}

}

\subsection{Flattening the Dendrogram}\label{sec:flattening}
\begin{algorithm}[!t]\caption{\flatten{}($D, t$)}\label{alg:flatten}
    \begin{algorithmic}[1]
    \Require{A dendrogram $D$ and a linkage similarity threshold $t$}
    \Ensure{A flat clustering}
    \State $\mathcal{C} := \{\}$ 
    \For{$d \in $ set of nodes of $D$}
       \If{the linkage similarity of $d \geq t$ and each ancestor of
       \indent $d$ has linkage similarity $< t$}
         \State Add the cluster corresponding to $d$ to $\mathcal{C}$
       \EndIf
    \EndFor
    \State \Return $\mathcal{C}$
    \end{algorithmic}
\end{algorithm}

While \terahac{} computes a dendrogram, numerous clustering applications require the algorithm to compute a \emph{flat} clustering, that is a partition of the input graph vertices.
A single dendrogram induces multiple flat clusterings of varying scales.
We now show how to \emph{flatten} the dendrogram, i.e., compute a canonical collection of flat clusterings consistent with it.

Let us assume that each internal node of the dendrogram is assigned the linkage similarity that was used to compute the corresponding cluster.
Let us also assume that the linkage similarity of each leaf node is infinite.
The algorithm that we use for flattening a dendrogram is given as Algorithm~\ref{alg:flatten}.
Given a threshold $t$, it picks each dendrogram node (cluster) which satisfies two conditions: (i) the linkage similarity of the node is at least $t$, and (ii) the linkage similarities of all ancestor nodes of $t$ are below $t$.
It is easy to see that this way we compute a flat clustering, that is each node is in exactly one cluster.

In a dendrogram returned by an exact HAC algorithm (i.e., a $1$-approximate one) the linkage similarites of the nodes along each leaf-to-root path form a nonincreasing sequence.
Hence, flattening a dendrogram using a threshold $t$ is equivalent to performing exactly the subset of merges described by the dendrogram whose linkage similarities are at least $t$.

This is not necessarily the case for $(1+\epsilon)$-approximate dendrograms.
However, we can still show that the linkage similarities along each leaf-to-root path are almost monotone, as shown in the following lemma.

\begin{lemma}\label{lem:flatten}
Let $\epsilon \geq 0$ and $D$ be a dendrogram obtained by performing $(1+\epsilon)$-good merges.
Assume we use Algorithm~\ref{alg:flatten} to flatten $D$ using parameter $t$.
Then, for each returned cluster the minimum linkage similarity of any merge used to create it is at least $t / (1+\epsilon)$.
\end{lemma}

\begin{proof}
Fix a cluster returned by Algorithm~\ref{alg:flatten}.
If the cluster has size 1, the lemma is trivial.
Otherwise, the cluster is obtained by merging two nodes $x$ and $y$.
Since the merge is $(1+\epsilon)$-good, at the point when the merge was performed
\[
\frac{\max(\wmax(x), \wmax(y))}{\min(\minmerge(x), \minmerge(y), w(xy))} \leq 1+\epsilon
\]
where $\min(\minmerge(x), \minmerge(y), w(xy)) = \minmerge(x \cup y)$ is exactly the minimum merge similarity of any merge used to build the cluster.
Moreover, Algorithm~\ref{alg:flatten} ensures $w(xy) \geq t$.
We have
\begin{align*}
\minmerge(x \cup y) & \geq \max(\wmax(x), \wmax(y)) / (1+\epsilon)\\ 
& \geq w(xy) / (1+\epsilon) \geq t/(1+\epsilon)
\end{align*}
%\vspace{-1em}
\end{proof}

\myparagraph{Vertex Pruning Optimization}
We now analyze the vertex pruning optimization (line \ref{l:optvertices} of Algorithm~\ref{alg:terahac}), which, after each round removes vertices whose highest incident edge weight is below $t/(1+\epsilon)$.
We show that vertex pruning with parameter $t'$ does not affect the final result, as long as we flatten the resulting dendrogram using a threshold $t \geq t'$.
At the same time, as we show in our empirical evaluation, vertex pruning brings significant performance benefits.

\begin{restatable}{lemma}{lemflat}
Let $G=(V,E,w)$ be a graph and $\epsilon \geq 0$.
Consider a run of Algorithm~\ref{alg:terahac} with a threshold parameter $t'$ followed by flattening the resulting dendrogram using a threshold $t$.
Then, the output of the algorithm is the same for any $t' \in [0, t]$.
\end{restatable}

\definecolor{light-gray}{gray}{0.95}
\definecolor{dark-gray}{gray}{0.25}

\begin{figure}[!ht]
\begin{minted}[fontsize=\footnotesize,linenos=true,numbers=left,xleftmargin=3em,escapeinside=||]{cpp}

KVTable<VertexId, DendrogramNode> TeraHAC(
    KVTable<VertexId, Vertex> graph, double eps,
    double t) {

// Initialize vertex metadata: set the cluster size to
// 1, and the min_merge_similarity to infinity.
graph = Initialize(graph);|\label{l:initialize}|

// Number of edges of weight at least t.
int64 num_edges = NumberOfHeavyEdges(graph, t);

// Stores the dendrogram nodes computed in each round
// of the while loop below.
std::vector<KVTable<VertexId, DendrogramNode>>
  dendrogram_nodes;

while(num_edges > 0) {
  // Partition the graph into clusters of <= 10M edges.
  // Computes a cluster id of each vertex.
  KVTable<VertexId, ClusterId> cluster_ids = 
    AffinityClustering(graph);
        
  // Join graph and cluster_ids on matching VertexId,
  // and key vertices by the cluster id.
  KVTable<ClusterId,
    pair<VertexId, Vertex>> clusters =
        KeyByClusterId(graph, cluster_ids);
    
  KVTable<VertexId, Dendrogram> new_nodes;
  KVTable<VertexId, ClusterId> sh_cluster_ids;
  KVTable<ClusterId, double> min_merge;
    
  // Run SubgraphHAC within each cluster.
  (new_nodes, sh_cluster_ids, min_merge) = 
    SubgraphHac(clusters.GroupByKey(), eps);
  dendrogram_nodes.push_back(new_nodes);
    
  // Vertex pruning: remove nodes, whose max incident
  // edge weight is below t/(1+eps).
  graph = graph.Apply(Prune(graph, t/(1+eps)))
     // Contract each cluster to a node.
     .Apply(Contract(graph, sh_cluster_ids))
     // Join the graph with min_merge on
     //matching keys.
     .Apply(SetMinMerge(min_merge))
     // Remove each vertex with no edges.
     .Apply(RemoveIsolatedVertices());
    
  num_edges = NumberOfHeavyEdges(graph, t);
}

return MergeDendrograms(dendrogram_nodes);
}
\end{minted}
\vspace{-1em}
\caption{\small \terahac{} pseudocode}
\label{fig:pseudocode}
\end{figure}

We acknowledge that the fact that \terahac{} requires the pruning optimization (which limits the output to the "bottom" part of the dendrogram) to achieve the best performance is a limitation of the algorithm.
At the same time, this limitation is also present in the existing methods providing similar scalability, that is affinity clustering~\cite{NIPS2017_2e1b24a6} and SCC~\cite{monath2021scalable}.
We note that on the datasets we studied in Section~\ref{sec:experiments} using this optimization makes very little impact on the clustering quality.

\section{\terahac{} Implementation}
\shepchange{We have implemented \terahac{} in Flume-C++~\cite{akidau2018streaming}, 
 an optimized distributed data-flow system similar to Apache Beam~\cite{beam}.}
A C++-based pseudocode for \terahac{} is given in \cref{fig:pseudocode}.
A \texttt{KVTable<K, V>} is a distributed collection of key-value pairs, where the keys have type \texttt{K} and values have type \texttt{V}~\cite{35650}.
In particular, the input graph to \terahac{} is a \texttt{KVTable<VertexId, Vertex>}.
Here, \texttt{VertexId} is a type used to represent unique ids of vertices and \texttt{Vertex} represents a vertex and its metadata.
Specifically, each vertex consists of 
\begin{itemize}
    \item the list of its incident edges, where each edge is represented as a pair of a \texttt{VertexId} specifying the edge endpoint and a floating-point number giving the weight,
    \item the size of the cluster represented by the vertex (initialized to $1$ on Line~\ref{l:initialize}),
    \item the minimum merge similarity used to create the cluster represented by the vertex (initialized to $\infty$ on Line~\ref{l:initialize}).
\end{itemize}

The result of the \terahac{} algorithm is a dendrogram represented as a \texttt{KVTable<VertexId, DendrogramNode>}. Each node is represented by
\texttt{DendrogramNode}, containing two fields: the id of the parent node in the dendrogram, and the linkage similarity of the merge that created the parent node.
Both fields are unset if the node represents a root in the dendrogram.

The algorithm runs in a loop, until the graph has no edges of weight at least $t$.
Each round first partitions the graph (see Line~\ref{l:partition} of Algorithm~\ref{alg:terahac}).
To this end we use the affinity clustering algorithm of Bateni et al.~\cite{NIPS2017_2e1b24a6}, which works as follows.
First, each vertex marks its highest weight incident edge.
Then, we find connected components spanned by the marked edges, which define the clusters.
Since each cluster computed by affinity clustering is later processed on a single machine, we use a size-constrained version of the algorithm, which ensures that each cluster contains at most 10 million edges~\cite{epasto2021clustering}.

The choice of affinity clustering is motivated by the following heuristic argument.
Intuitively, the goal is to use partitioning in which many edges inducing good merges (\cref{def:good}) have both endpoints in the same cluster.
From a vertex point of view its highest weight incident edge is most likely to induce a good merge (if we don't know its neighbors' metadata), and each such edge is an intra-cluster one in affinity clustering (unless the cluster is split to ensure the size constraint).

\myparagraph{SubgraphHAC}
Once we have found affinity clusters, we aggregate the vertices of each cluster on a single machine and run \subhac{} on each cluster.
This returns three results.
First, we obtain the set of new dendrogram nodes corresponding to the merges made by the \subhac{} algorithm (\texttt{new\_nodes}).
\revision{Second, \subhac{} computes the clustering (\texttt{sh\_cluster\_ids}) induced by performing all $(1+\epsilon)$-good merges made in this round.}
This clustering  maps the vertex ids of the graph in this round to cluster ids.
Third, for each cluster computed by \subhac{}, it computes the minimum merge similarity used to obtain that cluster.
For a fixed cluster returned by \subhac{}, the cluster is obtained by merging together some vertices, which correspond to clusters formed in the previous rounds of \terahac{}.
Hence, to compute the corresponding minimum merge similarity, we take the minimum of the similarities of the merges used to create the cluster in this  particular \subhac{} call, as well as the minimum merge similarities stored in the vertices in the cluster.
\subhac{} processes each cluster on a single machine, so all its three return values can be easily computed based on the vertices of the provided cluster, and the merges which were made.

\begin{table}\footnotesize
\centering
\centering
\caption{\small Graph inputs, including the number of vertices $(n)$, number of directed edges $(m)$, and the average degree $(m/n)$.
We show the statistics of the largest rMAT graph that we use.}
\vspace{-1em}
%\smallskip{}
\begin{tabular}[!t]{lrrr}   
\toprule
{Graph Dataset} & Num. Vertices & Num. Edges & Avg. Deg.\\
\midrule
{\emph{com-Orkut    } {\bf(OK)} }     & 3,072,627       & 234,370,166     & 76.2  \\
{\emph{Twitter      } {\bf(TW)} }     & 41,652,231      & 2,405,026,092   & 57.7  \\
{\emph{Friendster   } {\bf(FS)} }     & 65,608,366      & 3,612,134,270   & 55.0 \\
{\emph{rMAT-28} {\bf(RM28)} }         & 268,435,456     & 25,814,014,562  & 96.3 \\
{\emph{ClueWeb      } {\bf(CW)} }     & 978,408,098     & 74,744,358,622  & 76.3 \\
{\emph{Hyperlink} {\bf(HL)} }         & 3,563,602,789   & 225,840,663,232 & 63.3 \\
{\emph{Web-Query} {\bf(WQ)} }         & 31,764,801,992   & 8,670,265,361,938 & 272.9
\end{tabular}
\label{table:sizes}
\end{table}
\begin{figure*}[!t]
%\vspace{-1em}
%\vspace{1em}
\hspace{-1em}
\begin{minipage}{.71\columnwidth}
    \includegraphics[width=\columnwidth]{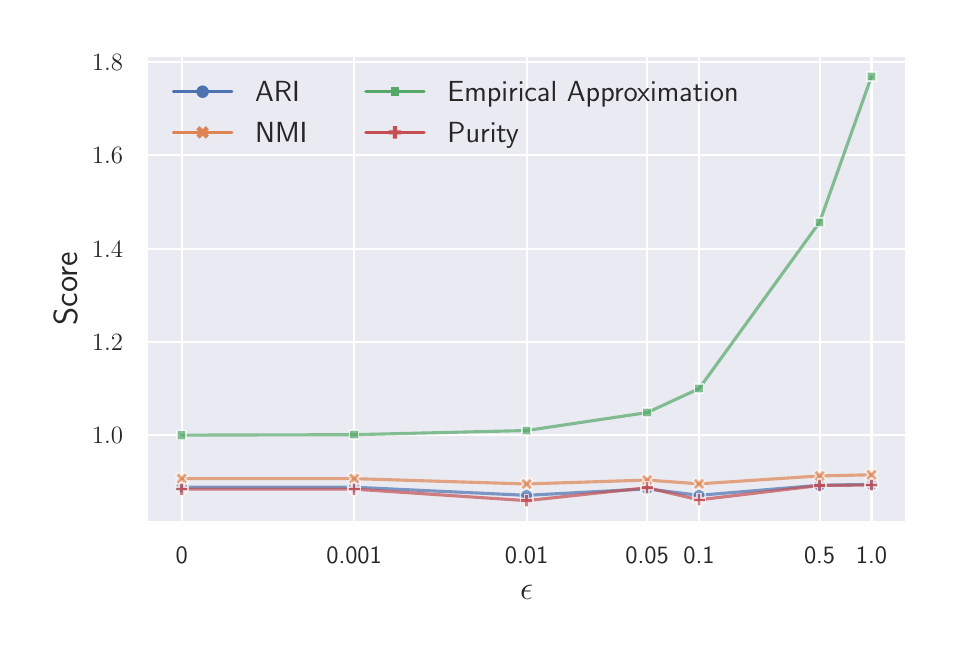}
  \vspace{-2em}
\end{minipage}%\hfill
\hspace{-1em}
\begin{minipage}{0.71\columnwidth}
  \includegraphics[width=\columnwidth]{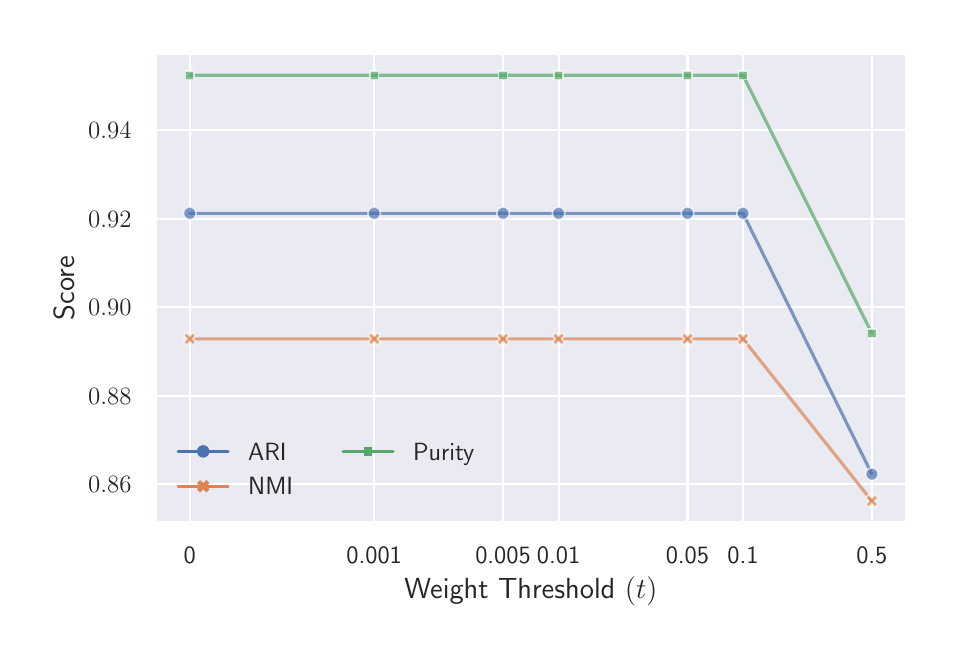}
  \vspace{-2em}
\end{minipage}
\begin{minipage}{.65\columnwidth}
\includegraphics[width=\columnwidth]{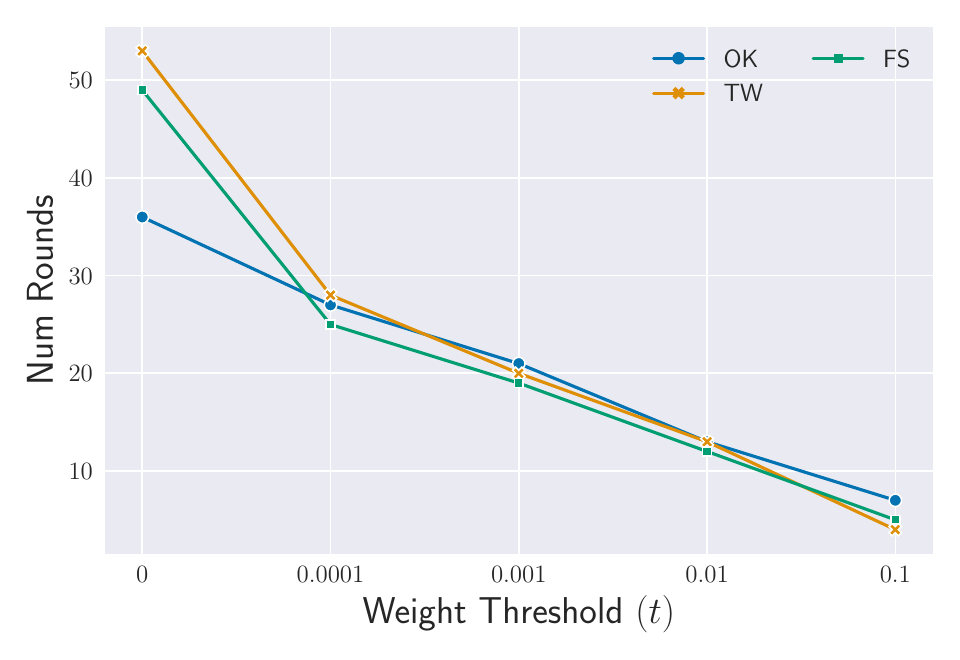}
  \vspace{-2em}
\end{minipage}\\
\begin{minipage}[t]{.6\columnwidth}
\vspace{-1em}
    \caption{
  \small
\revision{Quality of \terahac{}$_{\epsilon}^{t=0}$ run with varying $\epsilon$ on the 
Digits dataset graph constructed using $k=25$
(quality measures defined in Section~\ref{sec:quality}).}
\label{fig:digits-scores-eps-tune} 
  }
\end{minipage}\hfill
\begin{minipage}[t]{.75\columnwidth}
\vspace{-1em}
  \caption{
      \small
      Quality of \terahac{}$_{\epsilon=0.1}^{t}$ run
      with varying thresholds $t$ on the
Iris dataset graph constructed using $k=25$.
\label{fig:iris-scores-thresh-tune}
}
\end{minipage}\hfill
\begin{minipage}[t]{.65\columnwidth}
\vspace{-1em}
    \caption{
    \small
Num. rounds of \terahac{}$_{\epsilon=0.1}^{t}$ 
run with varying weight threshold $(t)$ on OK, TW, and FS.
Higher thresholds use fewer rounds.
\label{fig:threshold_vs_rounds}
  }
\end{minipage}
  %\vspace{-.5em}
\end{figure*}

\myparagraph{Vertex Pruning and Contraction}
To complete the round, we first apply vertex pruning, and update the graph based on the result of \subhac{}.
The key operation is to apply all the merges performed by  \subhac{} in this round.
We note that this is equivalent to contracting each cluster computed by \subhac{} to a single vertex.
This is what the \texttt{Contract} function does.
Assume that the vertex to cluster id mapping is given by a function $C$.
\texttt{Contract} first remaps each vertex id $x$ to $C(x)$, both in the keys of the graph \texttt{KVTable} and in the adjacency lists stored in \texttt{Vertex} objects.
This is achieved by two joins of \texttt{graph} with \texttt{cluster\_ids}.
Once all ids are remapped, vertices from the same cluster share the same \texttt{KVTable} key, and so we group them by key and merge together to obtain the final vertices.
Hence, the vertex ids in the new graph are the cluster ids computed by \subhac{}\footnote{This means that we use the same type to represent vertex and cluster ids. However, we decided to make the distinction to improve readability.}.

Next, we update the minimum merge similarity metadata, by joining the graph with the respective metadata computed by \subhac{}.
We finish updating the graph by applying an optimization: each vertex with no outgoing edges is removed, as it would clearly not participate in any merge anymore.

Our implementation is optimized to minimize the number of joins/shuffle steps.
In particular, while a standalone implementation of \texttt{Prune} and \texttt{SetMinMerge} would require a join/shuffle, in the implementation we actually extended the joins performed in \texttt{Contract} to perform these operations.

\revision{
\myparagraph{Shared-Memory Implementation}
To enable reproducibility, we also provide a shared-memory implementation of \terahac{}\footnote{\url{https://github.com/google/graph-mining/tree/main/in_memory/clustering/hac}}
Our shared-memory implementation of \terahac{} is written in the same parallel framework as \parhac{}, using the Parlay libraries for parallel primitives~\cite{blellochparlay} and the CPAM system for representing dynamic graphs~\cite{cpam}.
We provide (1) a faithful implementation of the size-constrained affinity algorithm~\cite{epasto2021clustering} used in our distributed implementation, (2) an efficient implementation of \subhac{} that runs in $O(m\log^2 n)$ time, and (3) an efficient implementation of weighted graph contraction to maintain the edge weights of the contracted graph after each round of partitioning and \subhac{}.

Although optimizing the running time of approximate HAC in the shared-memory setting was not the main goal of this paper, we found that our implementation of \terahac{} actually achieves consistent speedups over \parhac{}, even after optimizing \parhac{} to apply the vertex pruning optimizations used in \terahac{} (see Section~\ref{sec:scalability}).
}

\section{Empirical Evaluation}\label{sec:experiments}

\myparagraph{Graph Data}
We list information about graphs used in our experiments in
Table~\ref{table:sizes}. 
\defn{com-DBLP (DB)} is a co-authorship network sourced from the
DBLP computer science bibliography. 
\footnote{Source: \url{https://snap.stanford.edu/data/com-DBLP.html}.}
\defn{YouTube (YT)} is a social-network formed by 
user-defined groups on the YouTube site.
\footnote{Source: \url{https://snap.stanford.edu/data/com-Youtube.html}.}
\defn{LiveJournal (LJ)} is a directed graph of the social network.
\footnote{Source: \url{https://snap.stanford.edu/data/soc-LiveJournal1.html}.}
\defn{com-Orkut (OK)} is an undirected
graph of the Orkut social network.
\defn{Friendster (FS)} is an
undirected graph describing friendships from a gaming network.
Both graphs are sourced from the SNAP dataset~\cite{leskovec2014snap}.\footnote{Source: \url{https://snap.stanford.edu/data/}.}
%challenge~\cite{road-graph}.\footnote{http://www.dis.uniroma1.it/challenge9/}
\defn{Twitter (TW)} is a directed graph of the Twitter network, where 
edges represent the follower
relationship~\cite{kwak2010twitter}.
\footnote{Source: \url{http://law.di.unimi.it/webdata/twitter-2010/}.}
\defn{ClueWeb (CW)} is a web graph from the Lemur project at CMU~\cite{boldi2004webgraph}.
\footnote{Source: \url{https://law.di.unimi.it/webdata/clueweb12/}.}
\defn{Hyperlink (HL)} is a hyperlink graph obtained from the
WebDataCommons dataset where nodes represent web pages~\cite{meusel15hyperlink}.
\footnote{Source: \url{http://webdatacommons.org/hyperlinkgraph/}.}
The Web-Query graph consists of vertices representing 31 billion web queries, with
8.6 trillion edge weights corresponding to similarities between the queries (see Section~\ref{sec:webquery} for details).
Lastly, we run a set of scaling experiments on synthetic rMAT graphs~\cite{chakrabarti2004r},
which are constructed using the parameters $a = 0.6, b = c = 0.15, d = 0.1$.
The rMAT parameters were chosen to produce real-world-like graphs, following~\cite{wassington2022bias}.
rMAT-$X$ denotes an RMAT graph with $2^X$ nodes and $50 \cdot 2^X$ undirected edges (before removing duplicate edges).

We note that the large real-world graphs that we study are 
not weighted, and so we set the similarity of an edge $(u,v)$
to $\frac{1}{\log(\degree{u} + \degree{v})}$. We use this weighting
scheme since it promotes merging low-degree vertices over high-degree
vertices (unless there are sufficiently many edges); the same scheme
was also used by recent work on the topic~\cite{dhulipala2021hierarchical, parhac}.

\myparagraph{Building Similarity Graphs from Pointsets}
Our quality experiments run on graphs built from a pointset by 
computing the approximate nearest neighbors (ANN) of
each point, and converting the distances to similarities.
We convert distances to similarities
using the formula $\mathsf{sim}(u,v) = \frac{1}{1 + \mathsf{dist}(u,v)}$.
We then scale the similarities by dividing each similarity by the maximum similarity to ensure that the maximum similarity is $1$.
The same method of similarity graph construction was also used
in a recent paper on shared-memory parallel HAC~\cite{parhac}.

\myparagraph{Machine Configuration}
\shepchange{Our experiments are performed in a shared datacenter managed using Borg~\cite{49065}.}
Hence, our jobs compete for resources with
other jobs, and run alongside other jobs on the same machines.
We use a maximum of 100 machines to solve all problems
and a maximum of 800 hyper-threads across all machines, 
unless mentioned otherwise.
\shepchange{For details of the network topology, see~\cite{poutievski2022jupiter}.}
We ran each experiment three times, and find that the difference in running times across different trials was not significant (within 10\%), unless specified otherwise.

\myparagraph{Algorithms}
We compare \terahac{} with several HAC baselines.
We use \terahac$_{\epsilon=x}^{t=y}$ to denote \terahac{} run
with parameters $\epsilon=x$ and weight threshold $t=y$. When not
specified, \terahac{} denotes running the algorithm with $\epsilon=0.1$
and $t=0.01$, which are default parameter settings that we empirically 
find work well across a wide range of datasets.
Importantly, setting $\epsilon=0$ yields the exact HAC algorithm.
We also compare our algorithm with an optimized distributed 
implementation (written in the same framework) of the recently developed
SCC algorithm~\cite{monath2021scalable}, which we refer to as $\mathsf{SCC}$. 
The SCC algorithm is parameterized by $r$, the number of rounds used (denoted \scc-r). 
Increasing $r$ was observed to increase the quality of the algorithm~\cite{monath2021scalable}.
We evaluate a lower-quality (\scc-5), medium-quality (\scc-25) and
high-quality (\scc-100) setting.
We also (indirectly) compare our algorithm with \rac{} (e.g., in Fig.~\ref{fig:terahac_vs_parhac_rac_rounds}, discussed
in Section~\ref{sec:intro}).
Although we did not perform a running time comparison between \terahac{} and \rac{} in this paper, \terahac{} using $\epsilon=0$ should be strictly superior to \rac{} in practice and can be viewed essentially as an optimized version of the \rac{} algorithm. This is because for $\epsilon = 0$ \terahac{} performs $1$-good merges, which are exactly the merges that the RAC algorithm can perform (see Observation~\ref{obs:eqrac}).
However, contrary to RAC, \terahac{} may perform multiple merges involving the same vertex in a single round.

\begin{table*}\small
%\vspace{-4em}
\centering
\caption{\small Adjusted Rand-Index (ARI), Normalized Mutual
Information (NMI), Dendrogram Purity, and Dasgupta cost
of \terahac{} (columns 2--5) versus \scc-5, \scc-25, \scc-100 (columns 6--8), 
the HAC implementation of average linkage from SciPy (column 9), \revision{and the DBSCAN implementation from SciPy (column 10)}.
Note that \terahac{}$_{\epsilon=0}^{t=0}$ is the same as the \rac{} algorithm.
\revision{For DBSCAN, we perform a grid search for $\epsilon \in [0.01, 10000]$ and $\mathsf{minpts} \in [2, 128]$ and use the best scores that we obtain.}
All graph-based implementations are run
over the similarity graph constructed from an approximate $k$-NN graph with $k = 25$. 
The Dasgupta cost is computed over the complete similarity graph generated from
the all-pairs distance graph, and thus takes into account all pairwise similarities.
In the case of Dasgupta cost, lower values are better.
For the remaining measures, higher values are better.
We display the best quality score for each graph in green and underlined.
\revision{We show the relative improvement of each score over the score for \terahac{}$_{\epsilon=0}^{t=0}$ in parenthesis after each score.}}
\vspace{-1em}
\smallskip{}
\setlength{\tabcolsep}{2pt}
\begin{tabular}{@{}cl cccc|ccc | cc }
\toprule
& {Dataset} & \terahac{}$_{\epsilon=0}^{t=0}$ & \terahac{}$_{\epsilon=0}^{t=0.01}$ & \terahac{}$_{\epsilon=0.1}^{t=0}$ & \terahac{}$_{\epsilon=0.1}^{t=0.01}$ & \scc-5 & \scc-25 & \scc-100 & Sci-Avg & \revision{DBSCAN}\\
\midrule
\multirow{5}{*}{\STAB{\rotatebox[origin=c]{90}{{ARI}}}}

% ARI
& \emph{iris}   &  \best{0.92 (1x)} & \best{0.92 (1x)} & \best{0.92 (1x)} & \best{0.92 (1x)}    & 0.74 (.80x)  & 0.79 (.86x) & 0.87 (.94x)  & 0.75 (.82x) &  0.56 (.62x) \\
& \emph{wine}   &  0.37 (1x) & 0.37 (1x) & 0.37 (1x) & 0.37 (1x)   & 0.30 (.81x) & 0.22 (.61x) & 0.24 (.64x) & 0.35 (.94x) & \best{0.39 (1.05x)} \\
& \emph{digits} &  \best{0.88 (1x)} & \best{0.88 (1x)} & 0.87 (.99x) & 0.85 (.97x)    & 0.87 (.99x) & 0.83 (0.94x) & 0.82 (.93x)   & 0.69 (.78x) & 0.31 (.36x)  \\
& \emph{faces}  &  \best{0.57 (1x)} & \best{0.57 (1x)} & 0.56 (.98x) & 0.56 (.98x)    & 0.39 (.68x) & 0.51 (.89x) & 0.55  (.96x) & 0.52 (.92x) & 0.098 (.17x) \\

% NMI
\midrule
\multirow{5}{*}{\STAB{\rotatebox[origin=c]{90}{{NMI}}}}
& \emph{iris}   & \best{0.89 (1x)} & \best{0.89 (1x)} & \best{0.89 (1x)} & \best{0.89 (1x)} & 0.75 (.84x) & 0.77  (.86x) &  0.84 (.94x)   & 0.80 (.89x) & 0.73 (.82x)  \\
& \emph{wine}   & 0.42 (1x) & 0.41 (1x) & 0.42 (1x) & 0.42 (1x)   & 0.37 (.88x) & 0.38 (.90x) &  0.38 (.90x)   & 0.42 (1x) & \best{0.45 (1.07x)} \\
& \emph{digits} & \best{0.90 (1x)} & \best{0.90(1x)}  & 0.89 (.98x) & 0.89  (.98x)  & 0.90 (1x) & 0.87 (.96x) &  0.87 (.96x)  & 0.83 (.92x) & 0.66 (.73x) \\
& \emph{faces}  & \best{0.86 (1x)} & \best{0.86 (1x)} & \best{0.86 (1x)} & \best{0.86 (1x)}  & 0.83 & 0.85 &  \best{0.86 (1x)}   & \best{0.86 (1x)} & 0.68 \\

\midrule
\multirow{5}{*}{\STAB{\rotatebox[origin=c]{90}{{Purity}}}}
& \emph{iris}   &  0.94 (1x) & 0.94 (1x) & \best{0.95 (1.01x)} & \best{0.95 (1.01x)}   & -- & -- & --   & 0.86 (.90x) & --\\
& \emph{wine}   &  0.62 (1x) & 0.60 (.96x) & \best{0.62 (1x)} & \best{0.62 (1x)}   & -- & -- & --   & 0.62 (1x) & -- \\
& \emph{digits} &  \best{0.88 (1x)} & 0.87 (.98x)  & 0.87  (.98x) & 0.85 (.96x)   & -- & -- & --   & 0.75 (.85x) & -- \\
& \emph{faces}  &  0.61 (1x) & 0.61 (1x) & 0.58 (.95x) & 0.58 (.95x)   & -- & -- & --   & \best{0.62 (1.01x)}  & -- \\

\midrule
\multirow{5}{*}{\STAB{\rotatebox[origin=c]{90}{{Dasgupta}}}}
& \emph{iris}   & 321290 (1x)     & 321290 (1x)     & 320846 (1x)      & 320846 (1x)     & -- & -- & --   & \best{310957 (1.03x)} & -- \\
& \emph{wine}   & 26904 (1x)     & 28581 (.94x)     & \best{26902 (1x)}      & \best{26902 (1x)}       & -- & -- & --   & 27324 (.98x) & -- \\
& \emph{digits} & 243191685 (1x) & 243087881 (1x)  & 243323801 (.99x)  & 245791675 (.98x)   & -- & -- & --   &  \best{240476750 (1.01x)} & -- \\
& \emph{faces}  & 4631561 (1x)   & 4631561 (1x)    & 4627286 (1x)    & 4627286 (1x)    & -- & -- & --   &  \best{4569916 (1.01x)} & --

\end{tabular}
\label{table:qualitytable}
\end{table*}

\subsection{Quality}\label{sec:quality}
We start by understanding the quality of \terahac{} as a function of both $\epsilon$
and the weight threshold $t$. 
\revision{All results in this section can be reproduced using our shared-memory implementation.}
We compare the quality  
to known ground-truth clusterings for the
\emph{iris}, \emph{wine}, \emph{digits},
and \emph{faces} classification datasets
from the UCI dataset repository (found in the sklearn.datasets package). 
\shepchange{Unfortunately, we are not aware of other large-scale publicly
available datasets providing ground truth clustering labels.}
Our goal in this sub-section is to understand the largest
values of $\epsilon$ and $t$ that can be used without damaging accuracy. Note that in our
experiments, weights in the graph are always within $[0, 1]$.
To measure quality, we use the
\defn{Adjusted Rand-Index (ARI)} and 
\defn{Normalized Mutual Information (NMI)} scores,
which are standard measures of the quality of a clustering with respect
to a ground-truth clustering.
We also use the \defn{Dendrogram Purity} measure~\cite{heller05bayesian},
which takes on values
between $[0, 1]$ and takes on a value of $1$ if and only if the tree contains
the ground truth clustering as a tree consistent partition (i.e., each
class appears as exactly the leaves of some subtree of the tree).
Given a tree $T$ tree with leaves $V$, and a ground truth partition of
$V$ into $C=\{C_1, \ldots, C_l\}$ classes, define the purity 
of a subset $S \subseteq V$ with respect to a class $C_i$ to be $\mathcal{P}(S, C_i) = |S \cap C_i|/|S|$.
Then, the purity of $T$ is
\[
\mathcal{P}(T) = \frac{1}{|\emph{Pairs}|}\sum_{i=1}^{l} \sum_{x,y \in C_i, x\neq y} \mathcal{P}(\mathsf{lca}_{T}(x,y), C_i)
\]
where $\emph{Pairs}=
\{(x,y)\ |\ \exists i \text{ s.t. } \{x,y\} \subseteq C_i\}$ and
$\mathsf{lca}_T(x,y)$ is the set of leaves of the least common ancestor of $x$ and $y$ in $T$.
We also study the unsupervised \defn{Dasgupta Cost}~\cite{Dasgupta2016} measure of our dendrograms,
which is measured with respect to an underlying similarity graph $G(V, E, w)$
and is defined as:
$\sum_{(u,v) \in E} |\mathsf{lca}_{T}(u,v)| \cdot w(u,v)$.
The Dasgupta cost is computed over the complete similarity graph obtained
by computing all point-to-point distances.
\revision{Lastly,
given a dendrogram and its greedy merge sequence we compute for each merge
the ratio between the highest similarity edge in the graph and the similarity of the
merge, and take the maximum of these ratios to be the \defn{Empirical Approximation Ratio} (see Lemma~\ref{lem:empapx}).
For a $(1+\epsilon)$-approximate algorithm, the empirical approximation ratio is
at most $1+\epsilon$.
}

Before jumping in, we make a few observations.
If we use a weight threshold of $t=0$, the algorithm will compute a complete
$(1+\epsilon)$-HAC dendrogram (similar to \parhac{} or \seqhac{}), but as we shall
see in Section~\ref{sec:scalability}, computing the full $(1+\epsilon)$-dendrogram 
results in the size of the graph shrinking more slowly, and thus requires
more rounds.
On the other hand, using a weight threshold of $t > 0$ will result in the algorithm
potentially running much faster due to the size of the graph shrinking rapidly (as
we show in Section~\ref{sec:scalability}). However, from a quality perspective, a worry is
that in exchange for speed, we will sacrifice quality. In this sub-section, we will
understand how to set $\epsilon$ and $t$, and will demonstrate that this worry
is unfounded. We will show the following key points:
\begin{enumerate}[label=(\arabic*),topsep=0pt,itemsep=0pt,parsep=0pt,leftmargin=15pt]
  \item Without thresholding, values of $\epsilon \in [0, 0.5]$
  yield results that are typically within a few percent of the accuracy of 
  the exact algorithm.
  Based on our results, we choose a value of $\epsilon = 0.1$ for our experiments.
  
  \item Fixing $\epsilon=0.1$, values of $t \leq 0.01$ reliably
  achieve very high quality on our datasets with ground-truth.
\end{enumerate}

\begin{figure*}[!t]
%\vspace{-1em}
%\vspace{1em}
\hspace{-1em}
\begin{minipage}{.71\columnwidth}
    \includegraphics[width=\columnwidth]{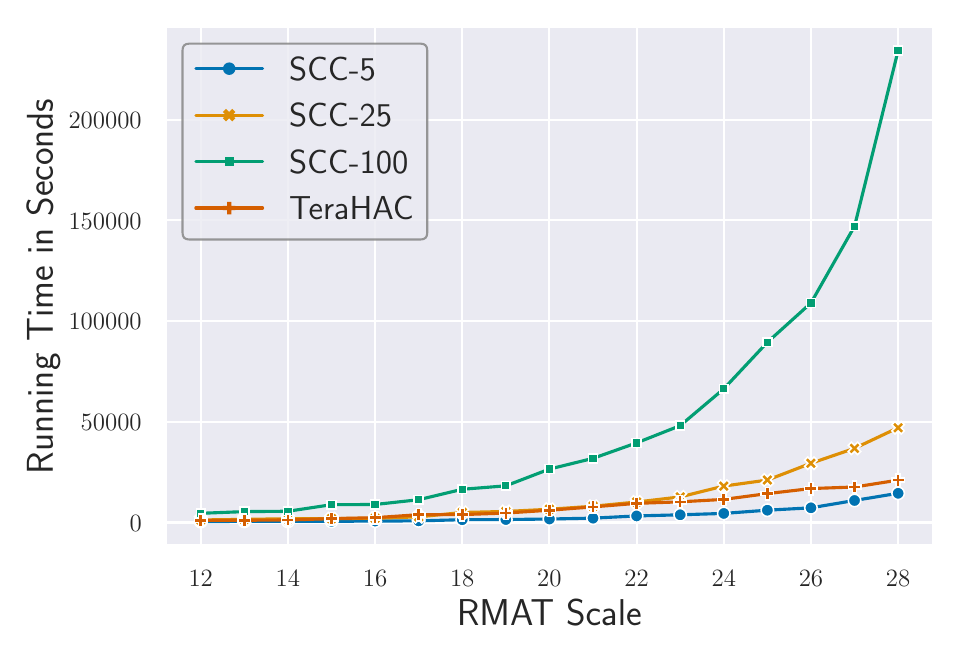}
  \vspace{-2em}
\end{minipage}%\hfill
%\hspace{-1em}
\begin{minipage}{0.71\columnwidth}
  \includegraphics[width=\columnwidth]{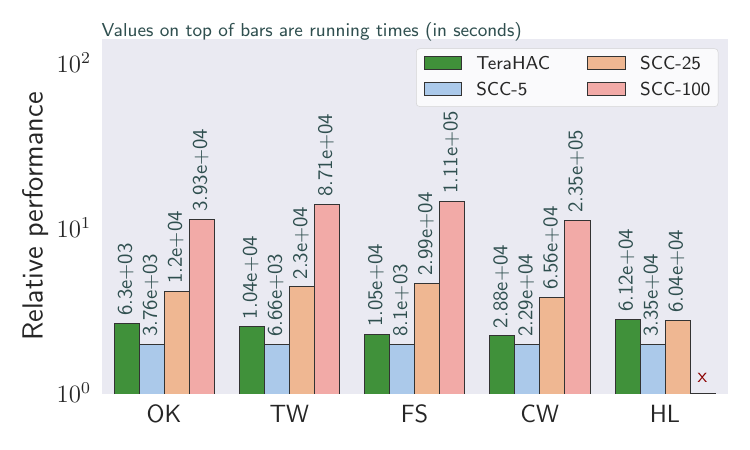}
  \vspace{-1em}
\end{minipage}
\begin{minipage}{.71\columnwidth}
  \includegraphics[width=\columnwidth]{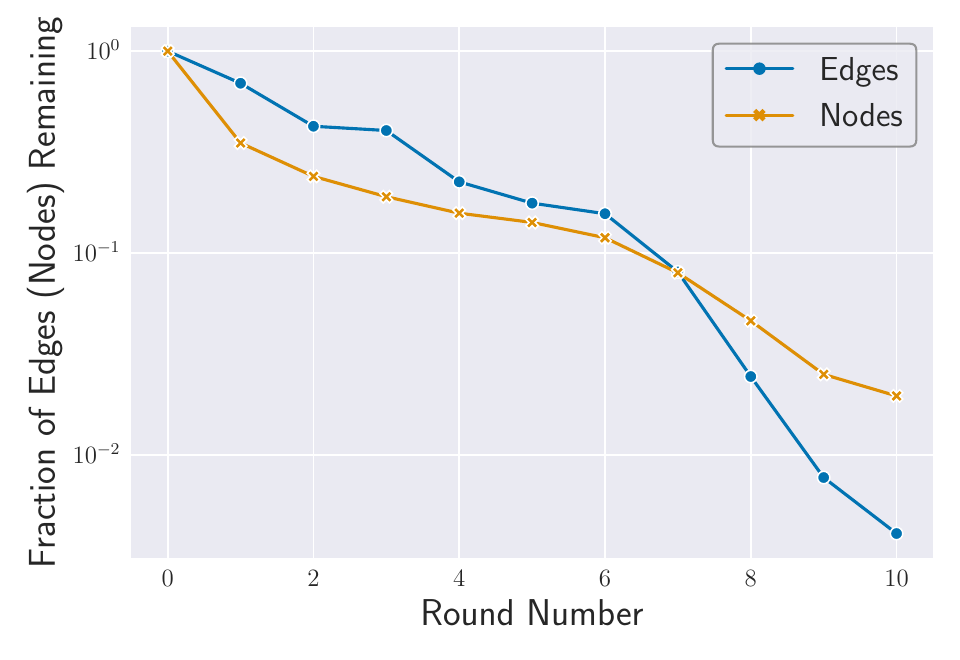}
  \vspace{-2em}
\end{minipage}\\
\begin{minipage}[t]{.6\columnwidth}
\vspace{-1em}
    \caption{
  \small
Running time in seconds of \terahac{}$_{\epsilon=0.1}^{t=0.01}$ and SCC using low, medium, and high-quality settings ($r=\{5, 25, 100\}, t=0.01$) on rMAT graphs of varying scales.
\label{fig:rmat-scaling}
  }
\end{minipage}\hfill
\begin{minipage}[t]{.75\columnwidth}
\vspace{-1em}
  \caption{
  \small
\revision{Relative performance (speedup over the fastest algorithm for each graph)} of \terahac{}$_{\epsilon=0.1}^{t=0.01}$ and SCC using low, medium, and high-quality settings ($r=\{5, 25, 100\}, t=0.01$) on real-world social and Web graphs.
\label{fig:real-world-graphs}
}
\end{minipage}\hfill
\begin{minipage}[t]{.65\columnwidth}
\vspace{-1em}
    \caption{
    \small
Decrease in number of edges and nodes as a function of the round number in \terahac{}$_{\epsilon=0.1}^{t=0.01}$ for rMAT-$28$.
\label{fig:rmat-edges-nodes}
    }
\end{minipage}
  %\vspace{-.5em}
\end{figure*}

\myparagraph{Tuning $\epsilon$}
We ran \terahac{} with varying values of $\epsilon$
without thresholding to understand what value of $\epsilon$ is sufficient
for achieving good quality relative to the exact baseline of $\epsilon = 0$.
The experiment runs \terahac{} on similarity graphs constructed 
from a $k$-NN graph of the
underlying pointset, with $k=25$, as described earlier in Section~\ref{sec:experiments}, and consistent with prior work~\cite{dhulipala2021hierarchical, parhac}.
We use all 4 UCI datasets mentioned earlier;
due to space constraints we show a representative result for the digits dataset 
in Fig.~\ref{fig:digits-scores-eps-tune} for all quality measures.
Perhaps surprisingly, Fig.~\ref{fig:digits-scores-eps-tune} shows that
the quality measures do not degrade significantly as $\epsilon$ is increased.
We noticed similar trends for other values of $k$; the reason is likely due to
the fact that in our experiments, \subhac{} makes very similar merges to $\epsilon=0$
even when using larger values of $\epsilon$.
We select the value $\epsilon=0.1$ to be consistent with the choice of $\epsilon$
in prior papers on approximate HAC~\cite{dhulipala2021hierarchical, parhac}.
Across all datasets, we found that for $\epsilon \leq 0.1$
on average the outputs are within 1.3\% of the best ARI score, 
within 0.25\% of the best NMI score, within 2.6\% of the best 
Purity score, and within 1.6\% of the best Dasgupta score.
Fig.~\ref{fig:digits-scores-eps-tune} empirically 
confirms our theoretical
result that the approximation factor is always within a factor of $(1+\epsilon)$.
For the remainder of the paper, we use a value of $\epsilon=0.1$ unless mentioned otherwise.

\myparagraph{Tuning the Threshold ($t$) with $\epsilon=0.1$}
Computing the complete dendrogram can be unnecessary if the flat clustering we seek uses a large threshold (i.e., does not merge edges with similarity below $t$). 
As Fig.~\ref{fig:threshold_vs_rounds} illustrates, larger thresholds require fewer rounds to compute, and can thus significantly improve the running time.
To understand how to select an appropriate threshold, $t$, we  run \terahac{} with varying values of $t$, while fixing $\epsilon=0.1$.
Fig.~\ref{fig:iris-scores-thresh-tune} shows a representative result on the iris dataset, showing essentially
no difference across $t \in [0, 0.1]$.
We tuned the threshold and find that for $t=0.01$, across all
datasets, the accuracy results are very close to that of $t=0$;
on average within 0.4\% of the ARI
score of $t=0$,  within 2.6\% of the Purity
score of $t=0$, and within 1.6\% of the Dasgupta cost of $t=0$.
We find that $t=0.01$ achieves the same NMI as $t=0$.
We set $t=0.01$ in the remainder of the paper.

\begin{figure}[t]
\begin{center}
\includegraphics[width=0.4\textwidth]{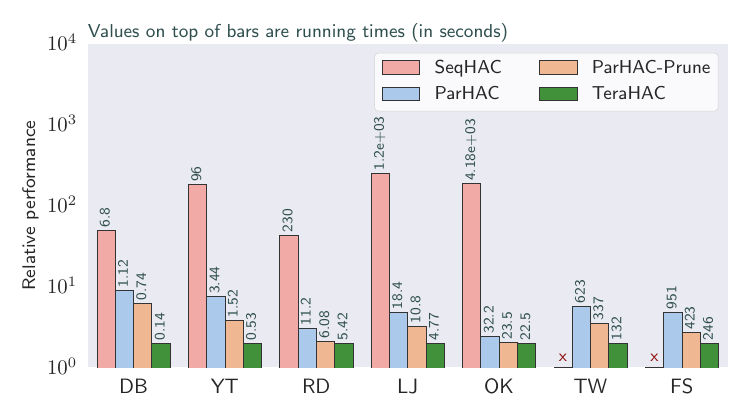}
\vspace{-1em}
\caption{\small \revision{Relative performance of \terahac{} vs. \parhac{} and \seqhac{}, all using $\epsilon=0.1$. \parhac{} is the original algorithm, and \parhac{}-Prune is the original algorithm using vertex pruning.
\label{fig:shared-memory-performance}}}
%\vspace{-1em}
\end{center}
%\vskip -0.1in
\end{figure}

\myparagraph{Quality Comparison}
Table~\ref{table:qualitytable} shows the results of our quality study when
comparing \terahac{} with different settings of $\epsilon$ and $t$ with the SCC
algorithm. 
The non-zero threshold we use $(t=0.01)$ is tuned as described above.
The SCC algorithm can vary the number of rounds, $r$, that is uses, which results
in $r$ flat clusterings.
We compute the quality of SCC by evaluating every flat clustering it produces and
using the best score across any clustering (i.e., for $r=25$ we evaluate all 25
flat clusterings).
\revision{We also compare with the exact $O(n^2)$ time average-linkage metric HAC implementation from 
sklearn~\cite{scipyhac} as well as the DBSCAN implementation from sklearn~\cite{scipydbscan}.
Both algorithms compute the full distance matrix and thus run in $O(n^2)$ time.
For DBSCAN, we set $\epsilon$ and $\mathsf{minpts}$ by searching over the range
$\epsilon \in [0.01, 10000]$ and $\mathsf{minpts} \in [2, 128]$.}
We consider four settings for \terahac{}; $\epsilon=0, t=0$, which yields an
exact HAC dendrogram that is equivalent to the \rac{} algorithm; $\epsilon=0, t=0.01$;
$\epsilon=0.1,t=0$; and lastly, using $\epsilon=0.1, t=0.01$.
All settings of \terahac{} achieve high-quality
results across all of the quality measures that we evaluate, across
all datasets.

Our results in Table~\ref{table:qualitytable} show that \terahac{} achieves
quality within a few percentage points of the best result on all datasets:
we find that for $\epsilon \leq 0.1$
on average the outputs are within 1.3\% of the best ARI score, 
within 0.25\% of the best NMI score, within 2.6\% of the best 
Purity score, and within 1.6\% of the best Dasgupta score.
The results for the \emph{exact} parameter settings of \terahac{}
show that thresholding using $t=0.01$ has little effect. 
We also observe that SCC is similarly affected, and this algorithm, which
is the prior state-of-the-art in scalable HAC, also requires careful selection of $t$.
SCC generally improves in quality as we increase the number of
compression rounds, with the exception of the digits dataset, where $r=5$ achieves
good quality. 
Comparing SCC to \terahac{},
\terahac{}$_{\epsilon=0.1}^{t=0.01}$ is superior to SCC with the exception of
the digits dataset, where \scc-5 achieves 0.4\% better ARI and 1.2\% better NMI.
Based on our results, we classify $r=\{5, 25, 100\}$ as low, medium, and high-quality
settings of SCC, respectively.
\revision{Lastly, compared to DBSCAN, for these datasets, HAC-based algorithms typically yield much higher quality. However, on the Wine dataset, DBSCAN finds clusterings that have up to 5.6\% higher ARI and 6.5\% higher NMI than TeraHAC (and also exact HAC).}

\subsection{Scalability}\label{sec:scalability}
Thus far, we have seen that the quality of \terahac{} with $\epsilon=0.1, t=0.01$ is
close to that of an exact HAC baseline. 
Next we evaluate the scalability of our algorithm under these settings.

\begin{figure*}[!t]
%\vspace{-1em}
%\vspace{1em}
\hspace{-1em}
\begin{minipage}{.71\columnwidth}
    \includegraphics[width=\columnwidth]{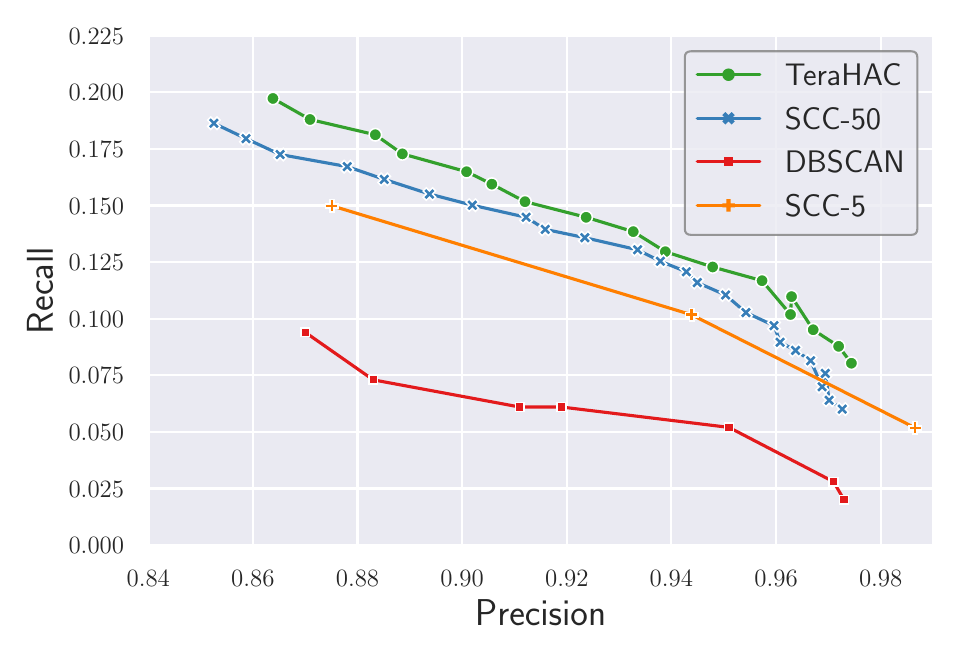}
  \vspace{-2em}
\end{minipage}%\hfill
\hspace{-1em}
\begin{minipage}{0.71\columnwidth}
  \includegraphics[width=\columnwidth]{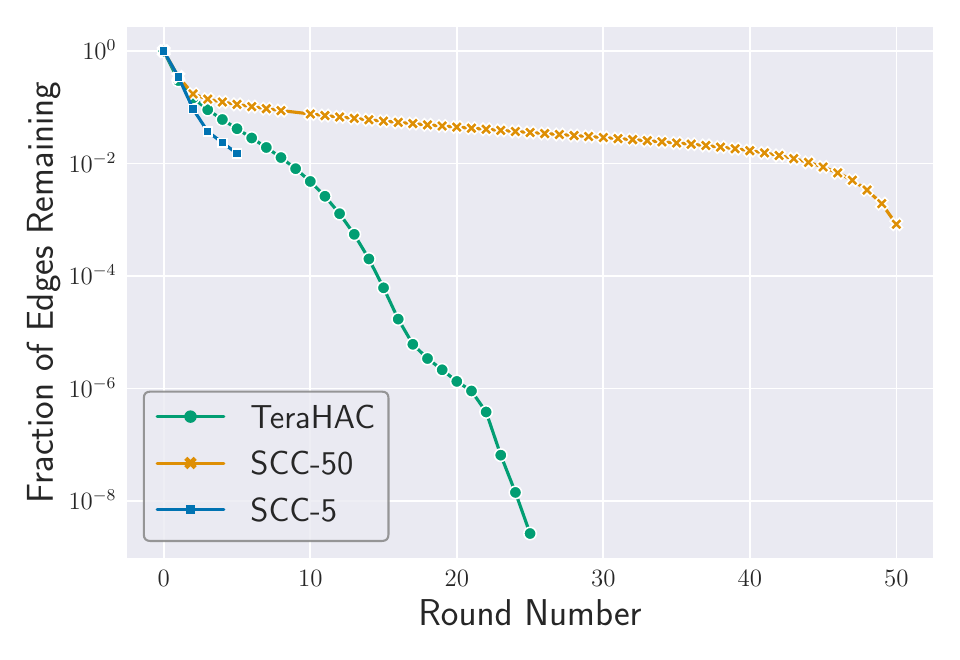}
  \vspace{-2em}
\end{minipage}
\begin{minipage}{.71\columnwidth}
  \includegraphics[width=\columnwidth]{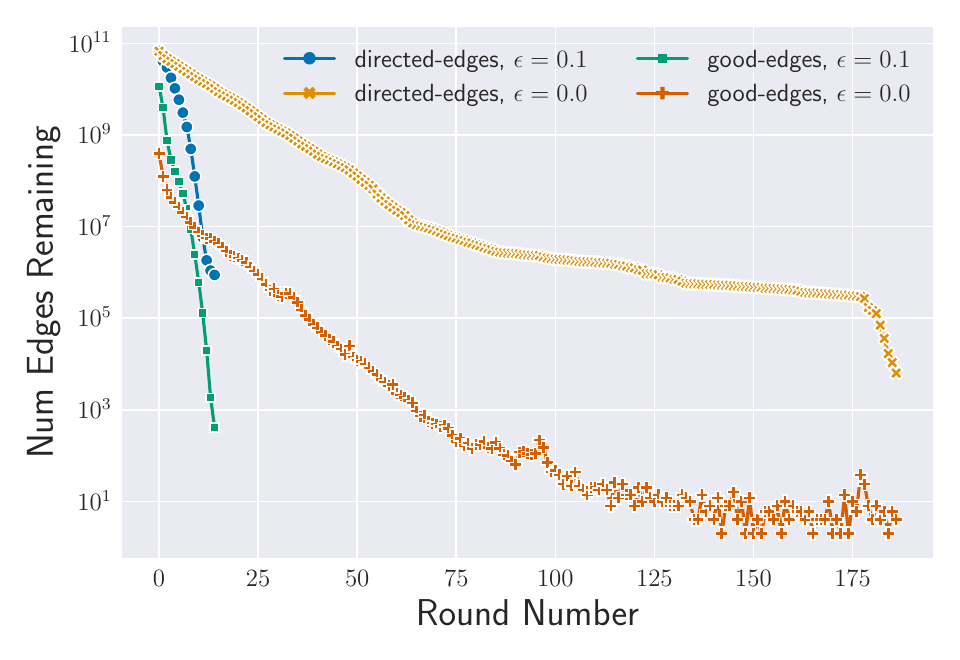}
  \vspace{-2em}
\end{minipage}\\
\begin{minipage}[t]{.6\columnwidth}
\vspace{-1em}
    \caption{
  \small
\revision{Precision-recall tradeoffs for \terahac{} and SCC on the large-scale web-query dataset.}
\label{fig:vasco-precision-recall}
  }
\end{minipage}\hfill
\begin{minipage}[t]{.75\columnwidth}
\vspace{-1em}
  \caption{
  \small
 \revision{Reduction in the number of edges for the web-query dataset. We note that the figure plotting the reduction in the number of nodes looks almost identical.}
\label{fig:vasco-edges}
}
\end{minipage}\hfill
\begin{minipage}[t]{.65\columnwidth}
\vspace{-1em}
\caption{\small 
Number of good and directed edges in each round 
on the CW graph ($n$=978M, $m$=74.7B).
The plots for other graphs show similar behavior when comparing \terahac{} with $\epsilon=0.1$ and $\epsilon=0$.
\label{fig:clueweb-good-edges}
%\revision{Reduction in the number of edges for the web-query dataset.}
%\label{fig:vasco-edges}
}
\end{minipage}
  %\vspace{-.5em}
\end{figure*}

%\begin{figure}[t]
%\begin{center}
%\includegraphics[width=0.36\textwidth]{figures/clueweb-good-edges.pdf}
%\vspace{-1em}
%\caption{\small Number of good edges and directed edges remaining
%as a function of rounds on the CW graph ($n$=978M, $m$=74.7B).
%The plots for other graphs show similar behavior when comparing \terahac{} with $\epsilon=0.1$ and $\epsilon=0$.
%\label{fig:clueweb-good-edges}}
%%\vspace{-1em}
%\end{center}
%%\vskip -0.1in
%\end{figure}

\myparagraph{Scalability on Large Graphs} 
We study the scalability of \terahac{} compared
with \scc{} on the rMAT graph family described at the start
of this section. 
The average-degree of these graphs is close to 100,
and is thus similar to our other real-world graphs (see Table~\ref{table:sizes}).
Fig.~\ref{fig:rmat-scaling} shows the result of our experiments. We observe that the
running time for \terahac{} lies between that of \scc{}-5 and \scc{}-25. For the largest 
rMAT graph that we study, which contains 268 million vertices and 25.8 billion edges,
\terahac{} is 11.4x faster than \scc{}-100, 2.23x faster than \scc{}-25, and only
1.45x slower than \scc-5.

Fig.~\ref{fig:real-world-graphs} shows the scalability of \terahac{}
compared with \scc{} on five large real-world graphs, including the largest
publicly available graph, the WebDataCommons Hyperlink graph (HL). Similar 
to our results on the rMAT graph family, we find that \terahac{} performs
between \scc-5 (0.67x as fast on average) and \scc-25 (2.04x as fast on average). The algorithm is significantly
faster than \scc-100 (8.3x faster on average), which is unable to finish within four days on HL.

Taken together with our results in Section~\ref{sec:quality},
\terahac{} is nearly as fast as the low-quality variant of \scc{}, while 
achieving much higher quality than the high-quality variant of \scc{}.

\myparagraph{Good Edges and Round-Complexity}
As our results in Fig.~\ref{fig:terahac_vs_parhac_rac_rounds} in Section~\ref{sec:intro}
show, \terahac{} uses several orders of magnitude fewer rounds than \rac{} and \parhac{}.
To better understand the reason behind the low round-complexity, we study the {\em number of good edges}
(i.e., $(1+\epsilon)$-good edges) available to the algorithm in each round.
This measure directly captures edges that can potentially participate in a merge.
Fig.~\ref{fig:clueweb-good-edges} shows a representative result for the CW graph.
The plots for other graphs are similar.
The number of rounds used by \terahac{} with $\epsilon=0.1$ is significantly lower 
than \terahac{} with $\epsilon=0$.
Initially, the number of good edges with $\epsilon=0.1$ is one order of
magnitude larger than with $\epsilon=0$.
Although \terahac{} with $\epsilon=0$ uses more rounds than with $\epsilon=0.1$, as shown
in Fig.~\ref{fig:terahac_vs_parhac_rac_rounds}, \terahac{} using $\epsilon=0$ still uses
many fewer rounds than \rac{}, although both algorithms solve HAC exactly.

\revision{
\myparagraph{Shared-Memory Performance} 
\terahac{} has the ability to significantly reduce the number of rounds required for approximately solving HAC at a given weight threshold.
We have found that our preliminary experiments with a shared-memory implementation of the algorithm show that its low round-complexity also translates into strong shared-memory performance.
We study the scalability of \terahac{} compared with state-of-the-art single-machine 
HAC implementations, namely \seqhac{}~\cite{dhulipala2021hierarchical} and \parhac{}~\cite{parhac} on a 72-core
Dell PowerEdge R930 (with two-way hyper-threading) with $4\times 2.4\mbox{GHz}$
processors and 1\mbox{TB} of main memory.
We show the results of our shared-memory experiment in Fig.~\ref{fig:shared-memory-performance}. \terahac{} is always
as fast or faster than all baselines, ranging from 1.45--8x speedup over \parhac{} and 48.5--185x speedup over \seqhac{}. We also obtain between 1.04--5.28x speedup over
a version of \parhac{} using the same vertex pruning optimization as \terahac{} before processing each bucket.
\terahac{} seems to have significant potential in the shared-memory setting and we plan to further investigate this direction in future work.
}

\subsection{Large Scale Clustering of Web Queries}\label{sec:webquery}

\begin{table}\footnotesize
\centering
\centering
\caption{\revision{\small Median running times on the web-query dataset.}}
\vspace{-1em}
\revision{
\begin{tabular}[!t]{lrrr}   
\toprule
{\terahac} & \scc-50 & \scc-5 & DBSCAN\\
\midrule
1280   & 2634       & 690  &  195 \\
\end{tabular}
}
\label{table:rtimes}
\end{table}

Lastly, we study the quality of \terahac{} on a large-scale dataset of web search queries.
We use a graph whose vertices represent queries, and edges connect queries of similar meaning.
The weight of each edge is computed using a machine-learned model based on BERT~\cite{devlin2018bert}.
The graph has about 31 billion vertices and $8$ trillion edges, and resembles the graph used in the evaluation of the SCC algorithm~\cite{monath2021scalable}.
We use a maximum of $48\,000$ cores across $6\,000$ machines.

For evaluation, we use a sample of $53\,659$ pairs of queries.
Each pair is assigned a human-generated binary label specifying whether the two queries likely carry the same intent and thus should belong to the same cluster.
$7104$ (about $13\%$) of the pairs have positive labels, and the remaining are negative.

In Fig.~\ref{fig:vasco-precision-recall} we show precision and recall with respect to these labels for \terahac{}, SCC \revision{and DBSCAN}.
We obtained our results by running \terahac{} with $t = 0.05$ and $\epsilon = 0.1$, and SCC with \revision{5 or} 50 rounds of compression decreasing the weight threshold down to $0.05$.
We evaluate different points on the precision-recall curve by considering different levels of clustering for SCC. For \terahac{}, different points are obtained by flattening the hierarchical clustering using different flattening thresholds.
The flattening algorithm uses batching to compute all flattenings simultaneously. Our implementation requires $140$ minutes to compute all clusterings we present in Fig.~\ref{fig:vasco-precision-recall}.

\revision{
The implementation of DBSCAN that we use is a natural adaptation of the DBSCAN algorithm (which in its vanilla version takes a set of points as an input) to a setting where the input is a similarity graph.
The algorithm takes two parameters: $\epsilon$ and $minPts$.
First, the algorithm considers each vertex a \emph{core} vertex if it contains at least $minPts$ incident edges of weight $\geq \epsilon$.
Then, we find connected components of the subgraph of the input graph consisting of core points and all edges of weight $\geq \epsilon$ between them (using the algorithm described in~\cite{lkacki2018connected}).
These components form core clusters.
Next, each non-core vertex which does not have a core vertex of similarity at least $\epsilon$ forms a singleton cluster.
Finally, for each remaining non-core vertex $v$, we assign $v$  to the cluster of its most similar core neighbor.
In Fig.~\ref{fig:vasco-precision-recall} we show results for $( \epsilon, minPts) \in \{(0.97, 128), (0.97, 256), (0.98, 64), \allowbreak (0.98, 96), (0.98, 128), \allowbreak (0.99, 32), (0.99, 64), (0.99, 128)\}$.
}

\revision{The median running times for the algorithms are given in Table~\ref{table:rtimes}.}
We find that \terahac{} achieves the highest recall at every value of precision in the range that we consider.
\revision{DBSCAN, while significantly faster than both \terahac{} and SCC, consistently obtains over 2x smaller recall than \terahac{}.}

\revision{
\terahac{} improves in quality over both \scc-5 and \scc-50, in particular delivering about 20\% better recall than \scc-5.
At the same time, it runs about 2x faster than \scc-50 and about 2x slower than \scc-5.
This difference in performance seems to be explained by reduction in the number of edges and nodes present in the graph for \terahac{}, as compared with SCC.}
Fig.~\ref{fig:vasco-edges} shows the reduction in the number of edges over the rounds of both algorithms.
For example, there are 8.6 trillion edges in this graph, and after 10 rounds, \scc-50 still has 561 billion edges remaining, whereas \terahac{} only has 41 billion edges remaining, which is 13.4x lower than \scc-50.
We observe a similar reduction in the number of nodes (clusters) that remain.
Both runs start with 31 billion nodes; after 10 rounds, 5.4 billion nodes remain for \scc-50, whereas only 799 million nodes remain for \terahac{}, which is 6.7x lower than \scc-50.
Our results show that even on extremely large real-world graphs, \terahac{} using conservative values of $\epsilon$ and weight threshold achieves excellent scalability relative to state-of-the-art distributed algorithms, and in fact out-performs these baselines. 
Crucially, these performance advantages are obtained while increasing accuracy at every point along the precision-recall tradeoff curve.

\section{Conclusion}

In this paper we introduced the \terahac{} algorithm and demonstrated its high quality and scalability on graphs of up to 8 trillion edges.
Our results indicate that \terahac{} may be the algorithm of choice for clustering large-scale graph datasets.

As a future work, it would be interesting to see whether we can theoretically bound the number of rounds required by the \terahac{} algorithm (possibly for a carefully chosen graph partitioning method).
Another open question is whether \terahac{} could be extended from computing the bottom (high-similarity) part of the dendrogram to computing the entire dendrogram.
Finally, we believe the notion of $(1+\epsilon)$-good merges may be useful for designing efficient HAC algorithms in other models, for example in the dynamic setting.

\bibliographystyle{plain}
\bibliography{references}

\clearpage{}
\appendix{}
\section{Missing Proofs}

\wmaxdecreases*
\begin{proof}
We prove the lemma by contradiction.
Pick the smallest $j \geq l$ such that $\wmax^{j}(v) < \wmax^{j+1}(v)$.
Then, $G_{j+1}$ is obtained from $G_j$ by merging some vertices $x$ and $y$, where $w(v, x \cup y) = \wmax^{j+1}(v)$.
By \cref{def:reducibility} we have $\wmax^{j+1}(v) = w(v, x \cup y) \leq \max(w(vx), w(vy)) \leq \wmax^j(v)$, a contradiction.
\end{proof}

\eqrac*
\begin{proof}
($\implies$) By the definition of $\wmax$ we have $w(uv) \leq \max(\wmax(u), \wmax(v))$.
Moreover, $w(uv) \geq \max(\wmax(u), \allowbreak \wmax(v))$ since, by the definition of $1$-good merge, we have $\max(\wmax(u), \wmax(v)) = \min(\minmerge(u), \minmerge(v), w(uv))$.

($\impliedby$) Since $w(uv) \leq \wmax(u)$ and $w(uv) \leq \wmax(v)$, $w(uv) = \max(\wmax(u), \wmax(v))$ implies $w(uv) = \wmax(u) = \wmax(v)$.
By Lemma~\ref{lem:goodinvariant}, $M(v) \geq \wmax(v)$ and $M(v) \geq \wmax(v)$.
This implies that the merge is $1$-good.
\end{proof}

\clmerge*
\begin{proof}
Consider the sequence of merges that were made to obtain the dendrogram $D$ and let $G'_1, \ldots, G'_n$ be the corresponding graphs.
Assume that the merge of $x'$ and $x$ was done in $G'_i$ (and resulted in obtaining graph $G'_{i+1}$).
Let $\wmax'(v)$ be the value of $\wmax(v)$ in graph $G'_i$.
Graph $G'_i$ does not necessarily have a vertex $y$, but since we assumed that the merge of $x'$ and $x$ happens before the merge of $y'$ and $y$, $G'_i$ contains a set of vertices $\{y_1, \ldots, y_t\}$, such that $\bigcup_{i=1}^t y_i = y$.
By \cref{def:reducibility} we have that $w(xy) = w(x, \bigcup_{i=1}^t y_i) \leq \max_{i=1}^t w(x, y_i) \leq \wmax'(x)$, which implies:
\[
\frac{w(xy)}{\min(\minmerge(x), \minmerge(x'), w(xx'))} \leq \frac{\max(\wmax'(x), \wmax'(x'))}{\min(\minmerge(x), \minmerge(x'), w(xx'))} \leq 1+\epsilon.
\]
Here, the second inequality follows from the fact that the merge of $x$ and $x'$ in $G'_i$ is good.
\end{proof}

\lemsub*
\begin{proof}
We observe that whether a merge is good only depends on the vertices involved in a merge, their $\minmerge$ values and their incident edges.
Hence, the lemma easily follows when $k=1$.
Let us now consider the case of $k = 2$, the argument can be easily continued inductively.
Let $C_1$ be the set of vertices obtained from $C$ by performing merge $m_1$, and let $G_1$ be the graph obtained from $G$ by performing merge $m_1$.
It is easy to see that by performing merge $m_1$ on $G^C$ we obtain the graph $G_1^{C_1}$.
That is, the graph that \subhac{} sees after applying the first merge, is the same graph that we would obtain if we applied the merge on the global graph, and then computed the input to \subhac{}.
The lemma follows.
\end{proof}

\lemflat*

\begin{proof}
By \cref{lem:flatten}, flattening the dendrogram with threshold $t$ only uses nodes of the dendrogram whose linkage similarity is at least $t / (1+\epsilon)$.
Pruning removes vertices whose maximum incident edge has weight strictly less than $t' / (1+\epsilon) \leq t / (1+\epsilon)$.
By \cref{def:reducibility}, each merge that such vertex participates in has  merge of similarity below $t / (1+\epsilon)$.
Moreover, the removal of such vertices does not affect whether a merge of similarity at least $t / (1+\epsilon)$ is good, or does not affect the edge weights participating in such merges.
The lemma follows.
\end{proof}

\section{\subhac{}}

In this section, we give a near-linear time algorithm for
computing a set of $(1+\epsilon)$-good merges
by dynamically maintaining the goodness
values of the edges, which we call \subhac{}. 
The algorithm has the following specification: it performs only
$(1+\epsilon)$-good merges, and ensures that upon completion,
there are no $(1+\epsilon')$-good merges for $\epsilon' = O(\epsilon)$ (see Lemma~\ref{lem:goodnessrange}).
Recall that \subhac{} is run on the 
partitions (i.e., subgraphs) of the input graph, where
vertices assigned to this partition are called 
\emph{active} and the neighbors of these vertices that
are not in this partition are \emph{inactive}. Furthermore,
the vertices carry their min-merge values from prior
rounds, $\minmerge{}(v)$.
In what follows, we use vertex and cluster interchangeably.
The input to \subhac{} is therefore a graph where vertices are
marked active and inactive, where each vertex $v$ has a min-merge value
$\minmerge(v)$, as well as a corresponding cluster size.
The goal of the algorithm is to merge a subset of the good edges, and run in near-linear time.

\myparagraph{Challenges}
Designing a near-linear time algorithm for this problem is surprisingly non-trivial.
One challenge is that the goodness of a vertex $v$'s incident edges 
can change significantly even without $v$ participating in merges
itself. For example, this could occur is if a neighbor of $v$ 
merges with another vertex, resulting in the neighboring vertex
increasing in size, causing $w_{uv}$ to decrease. A simple example, e.g.,
a sequence of $n-1$ merges of the degree-1 vertices of a star
into the center of the star illustrates that exact maintenance of
goodness values requires $O(n^2)$ work.
Another challenge is that the goodness of all edges incident to a vertex $v$ depends on the $\wmax(u)$ values for all neighbors of $v$, which in turn depends on the cluster sizes of their neighbors.
Hence, the goodness of an edge incident to $v$ may change when a vertex / cluster \emph{two hops away} from $v$ chages its size.

These examples serve to illustrate that attempting to {\em exactly} 
maintain the goodness values is probably hopeless.
One could ask why bother with goodness---perhaps the algorithm
can simply keep a heap induced only on the active edges and their weights?
The prior near-linear time $(1+\epsilon)$-approximate HAC 
algorithm~\cite{dhulipala2021hierarchical} used such an approach (this was
called the heap-based approach in the paper). 
However, the heap-based approach alone as used in \seqhac{} is insufficient for
our needs, since \seqhac{}  cannot merge any edge that has weight smaller
than a factor of $(1+\epsilon)$ from the current maximum weight in the graph.
Still, our algorithm borrows and builds on some ideas from the heap-based
approach. For example, it indexes a heap over the vertices, using the lowest 
(i.e., best) goodness as the priority for each vertex. It also maintains
weights using partial weights, where an edge to a neighbor $x$ of a vertex $v$ 
is stored as a {\em partial average-linkage weight}, $w(v,x) \cdot |v|$, i.e. the
average-linkage weight between vertices $(v,x)$ but multiplied by $v$, the source endpoint's cluster size.
However, many new ideas are required to efficiently
maintain the goodness values and argue that the algorithm ensures a $(1+\epsilon)$-approximation.

\newcommand{\exactgoodness}[0]{\ensuremath{g}}
\newcommand{\apxgoodness}[0]{\ensuremath{\tilde{g}}}
\newcommand{\curgoodness}[0]{\ensuremath{\ddot{g}}}
\newcommand{\weight}[0]{\ensuremath{w}}
\newcommand{\apxweight}[0]{\ensuremath{\tilde{w}}}

\myparagraph{Our Approach}
We design an algorithm based on a {\em lazy heap-based approach}, which
we give some high-level pseudocode for in Algorithm~\ref{alg:subhac_app}.
The algorithm maintains a priority queue (min heap) over the active vertices,
with a vertex's priority corresponding to the goodness of its smallest goodness neighbor.
The algorithm cannot maintain these goodness values exactly, but instead uses
approximate goodness values, $\apxgoodness$.
There are two sources of approximation that play a role. 
First, similar to the \seqhac{} algorithm, the algorithm maintains $(1+\alpha)$
approximations of the edge weights, $\apxweight(u,v)$,
by ``broadcasting'' a vertex's cluster size to its
neighbors and updating the weights based on the next cluster size.
Here, $\alpha \geq 0$ is a parameter of \subhac{}.
This broadcast
operation is done only when a vertex's cluster size increases by a $(1+\alpha)$ factor.
Second, it maintains an approximate view of the $\wmax(v)$ values for each vertex $v$, $\apxwmax(v)$.
We assign every active edge (an edge between two active endpoints) 
to the endpoint with larger $\apxwmax$ value; call this the {\em assigned endpoint} for an edge.
Each time the $\apxwmax$ value of a vertex changes by a $(1+\alpha)$ factor, 
the algorithm goes over all assigned edges and reassigns them to whichever endpoint has higher $\apxwmax$ value.
This approach yields good guarantees on (1) the approximation
quality of each edge we merge and (2) termination conditions for the algorithm (i.e., when
it terminates, no edges with goodness below a given threshold remain).
The algorithm itself (Algorithm~\ref{alg:subhac_app}) is relatively
simple, with most of the work being done by the graph
representation, described next.

\begin{algorithm}[!t]\caption{\subhac{}($G=(V, E, w, \minmerge), A, \epsilon, \alpha$)}\label{alg:subhac_app}
    \begin{algorithmic}[1]
    \Require{$G^A = (V, E, w, \minmerge)$, $A \subseteq V$, $\epsilon > 0$.}
    \Ensure{A set of merges of vertices in $A$}
    \State Mark vertices of $V \setminus A$ as \emph{inactive}
    \State For each active edge $(u,v)$ assign it to the endpoint with higher $\wmax$ value.
    \State $T := $ priority queue on active vertices of $G[A]$ keyed by the goodness of their lowest-goodness assigned edge.
    \While{$T$ is not yet empty}
        \State $(u, (v, b_{uv})) :=  \textsc{RemoveMin}(T)$
        \If {$b_{uv} > (1+\epsilon)$}
            \State $(u, v', b_{uv'}) := $ $\min$ goodness assigned edge of $u$
            \State Reinsert $(u, (v', b_{uv'}))$ to $T$ if $b_{uv'} \leq (1+\epsilon)$.
        \Else
            \State Merge $u$ and $v$ in $G$
        \EndIf
    \EndWhile
    \State Return all merges made
    \end{algorithmic}
\end{algorithm}

\subsection*{Graph Representation}
The graph representation takes as input parameters $\epsilon \geq 0, \alpha \geq 0$;
$\epsilon$ controls the accuracy and ensures that every merge 
is $(1+\epsilon)$-good; $\alpha$ controls the frequency with which 
the algorithm performs broadcasts, and affects the range of goodness
for which edges are guaranteed to be merged.

The graph representation
(1) support queries for the approximate goodness value $\apxgoodness(u,v)$ of each edge,
(2) support queries for the $\apxwmax(v)$ value for each vertex,
(3) support efficiently merging pairs of active vertex, and
(4) support querying for (approximately) the lowest goodness edge incident to a vertex $v$.
It supports (1--4) above
with an overall running time of all operations over any sequence 
of merges in $\tilde{O}(m + n)$ time.
We can support (4) only for edges with true goodness value in a given
range $[1, T]$, and elaborate more on this in Lemma~\ref{lem:goodnessrange}.

\myparagraph{Data Structures and Initialization}
Each active edge $(u,v)$ (an edge where both $u$ and $v$ are active)
is initially {\em assigned} to whichever of its endpoints has larger $\apxwmax$ value.
Let $A(u,v)$ be the endpoint that the edge $(u,v)$ is assigned to.
Initially, the $\apxwmax$ values are exact, and can be computed by
scanning all the edges.
For each active $v \in V$, we store all edges assigned to the vertex
in a priority queue $A(v)$, which is keyed by the approximate goodness
of the edge (from smallest to largest).
Initially, the goodness value for an assigned edge $(u,v)$ is $\apxgoodness(u,v)$.
We maintain for each active vertex $v$ its full neighborhood,
including both active and inactive neighbors. 
$N(v)$ stores the weights of {\em all} of its incident edges
using the one-sided partial weight representation: for each $(u,v)$ edge incident
to $u$, we store $w(u,v) \cdot |u|$, i.e., the true average-linkage weight, $w(u,v)$ multiplied by the
cluster size of $u$.
The neighborhood is stored as a sparse set, e.g., using a hashtable.
Each active vertex $u$ stores two quantities: 
(a) the initial cluster size of the vertex and
(b) the approximate $\wmax$ value of the vertex.
The stored values are updated whenever the cluster size of the vertex
increases sufficiently, or whenever the $\wmax{}$ value
changes sufficiently, as we discuss shortly.
Lastly, we maintain a list of edges $R$ to be {\em reassigned}
which is cleared at the end of each merge.

\myparagraph{Invariants}
The invariants are motivated by the following two broadcast
operations, which update information along all neighboring
edges whenever a vertex's cluster size or $\apxwmax$ value 
change by a $(1+\alpha)$ factor.
The first type of broadcast updates a node's cluster size in 
all of its neighbor's data structures after its cluster size increases
by a $(1+\alpha)$-factor.
The second type of broadcast is based on the $\apxwmax$ values: if
the $\apxwmax(u)$ drops by a $(1+\alpha)$ factor, we reassign all 
edges assigned to $u$.

We first state a useful invariant about the approximate weights, $\apxweight$.
Before any merges are performed, the weights are exact.
There may be three things affecting $w(u,v)$: the cut-size changing, or the
cluster size of either $u$ or $v$ changing. Changes of the first type are easy to
handle, and can update the weight of the edge to be exact. Changing the cluster size
of the endpoint to which the edge is assigned to (say $u$) is also easy, since the
weight is implicitly normalized by $|u|$. So the only error comes from changes to the
neighbor cluster size, $v$, but since we broadcast and make the incident edge weights exact
after a cluster grows in size by a $(1+\alpha)$ factor, the edge weight can be at most a $(1+\alpha)$ factor larger than the true weight.
This guarantee for the approximate weights, $\apxweight$, provided by the broadcasting procedure is summarized by the next invariant:
\begin{invariant}\label{inv:weight}
For every edge $(u,v) \in E,\weight(u,v) \leq \apxweight(u,v) \leq \weight(u,v)(1+\alpha)$.
\end{invariant} 
\noindent In other words, the approximate weights maintained for each edge are an upper-bound
on the true edge weight. Note that we can compute the {\em exact weight} of any edge in
$O(1)$ at any given time since the cut weight, and cluster sizes are exact; we need the approximate weights and \cref{inv:weight}
only when computing $\wmax(u)$, since we cannot 
afford to scan over all incident edges to $u$ to compute $\wmax$ exactly.

Let $\curgoodness(u,v)$ be the {\em stored goodness} of the edge $(u,v)$, i.e. the
stored goodness value in the priority queue $A(u)$ if $(u,v)$ is assigned to $u$.
At initialization, $\curgoodness = \apxgoodness = \exactgoodness$. However, as merges start
to occur in the graph, the stored goodness values, $\curgoodness$ may become approximate. Understanding
the stored goodness values is necessary, since when scanning the assigned edges of a vertex to find a potentially
mergeable edge, we use the stored goodness values, $\curgoodness$. The next invariant summarizes the relationship between
the $\curgoodness$ values and the exact goodness at any point in time in the algorithm (we provide the proof that the
invariant is maintained in Lemma~\ref{lem:merge}).
\begin{invariant}\label{inv:good}
For every edge $(u,v) \in E$, $\curgoodness(u,v) \leq \exactgoodness(u,v) (1 + \alpha)^2$
\end{invariant}

\noindent So at any point in time during the algorithm, if the goodness value of an edge
is small, then the stored goodness value of the edge is also small. Again, we emphasize the difference
between the three types of goodness values: $\exactgoodness$ is the true goodness value of an edge;
$\apxgoodness$ is a $(1+\alpha)$ over-approximation of the $\exactgoodness$ value that we can compute
at any instant; finally, $\curgoodness$ is the goodness value that is stored in the priority queue of
the assigned endpoint.
We further note that the invariant says nothing about a lower bound on $\curgoodness(u,v)$ in terms of $\exactgoodness(u,v)$. 
For example, it could be
the case that $\exactgoodness(u,v) \gg (1+\epsilon)$ while $\curgoodness(u,v) \leq (1+\epsilon)$,
but this is not an issue, as we can amortize the cost of checking the edge against the
increase in goodness for the edge.

Initially,
\cref{inv:weight} is satisfied since the weights are exact.
Lastly, \cref{inv:good} is satisfied initially since the goodness values are exact.

\myparagraph{Useful Lemmas}
Next, we state some useful lemmas about the approximation error
of computing $\apxwmax$ and $\apxgoodness$ at a given point in time,
based on the invariants above.

\begin{lemma}\label{lem:wmax}
$\forall u \in V$, $\wmax(u) \leq \apxwmax(u) \leq \wmax(u)(1+\alpha)$.
\end{lemma}
\noindent In other words, the approximate $\apxwmax$ values s a $(1+\alpha)$ over-approximation of the true $\wmax$ values.
\begin{proof}
Since we store a map indexed on the approximate partial weights of the neighbors and sorted
by the approximate edge weights $\apxweight$, we can simply use the largest $\apxweight$ of an
incident edge to be $\apxwmax(u)$.
The returned weight is a $(1+\alpha)$ approximation to $\wmax{}$ by Invariant~\ref{inv:weight}.
\end{proof}

Using \cref{lem:wmax}, we obtain guarantees on calculating a 
particular edge's goodness value. For the calculation, we use the
exact value of the edge weight, but suffer an approximation in the numerator;
when computing $\max(\apxwmax(u), \apxwmax(v))$, the value can be at most
a multiplicative $(1+\alpha)$ factor larger than the true value. This
is summarized in the next lemma:
\begin{lemma}\label{lem:good}
$\forall (u,v) \in E$, $\exactgoodness(u,v) \leq \apxgoodness(u,v) \leq \exactgoodness(u,v)(1+\alpha)$.
\end{lemma}
\begin{proof}
The approximate goodness of an edge is $U/L$ where
$U = \max(\apxwmax(u), \apxwmax(v))$ and $L = \min(\minmerge(u), \minmerge(v), \weight(u,v))$.
Since $U$ is a $(1+\alpha)$ approximation of the true numerator and cannot
be smaller, and $L$ is exact, the worst case is when $U$ is a $(1+\alpha)$ approximation.
Therefore, we have $\exactgoodness(u,v) \leq \apxgoodness(u,v) \leq (1+\alpha)\exactgoodness(u,v)$.
\end{proof}

We describe how $\apxweight$ and $\apxgoodness$ are calculated,
and why the lemmas are true in more detail below.
\cref{lem:good} is crucial since by only merging edges where
$\apxgoodness(u,v) \leq (1+\epsilon)$, the edge must truly be $(1+\epsilon)$-good.

\subsection*{Graph Representation Operations}

\myparagraph{Merge$(u,v)$}
We merge the vertex with smaller cluster size into the vertex with larger
cluster size. Suppose we are
merging $u$ into $v$ (thus $v$ will still be active after the merge).
Both $u$ and $v$ are initially active.
Merging the weights can be done in $O(|N(u)|)$ time by iterating over
edges in $N(u)$ and updating weights in $N(v)$ using the average-linkage
formula.
After merging the neighborhoods, we push all edges in $A(u)$
(edges assigned to $u$) to $R$ to be fixed at the end of the merge.
Finally, we see if we need to perform a {\em broadcast} operation due to
either the cluster size increasing sufficiently, or the $\wmax{}$ value
decreasing sufficiently:
\begin{enumerate}
\item If $v$'s cluster size grows larger by a $(1+\alpha)$ factor (from the last broadcast value), we iterate over all edges in $N(v)$ and update their weights in both $v$ and $v$'s neighbor's queues.\label{bcast:size}
\item If $\apxwmax(v)$ either grows or decreases by a $(1+\alpha)$ factor from the last value of $\apxwmax(v)$ that a broadcast occurred at, we iterate over all of $v$'s edges and reassign them to the endpoint with larger $\apxwmax$ value.\label{bcast:wmac}
\end{enumerate}
This concludes the description of the merge procedure.
We describe the correctness details (i.e., showing that the invariants
are preserved after this operation) shortly.

\myparagraph{BestAssignedNeighbor$(u)$}
This operation is used to compute the min goodness assigned neighbor of $u$
in Algorithm~\ref{alg:subhac_app}.
Similar to $\wmax{}$, getting the {\em true} best assigned neighbor would
require scanning all edges assigned to $u$.
Instead, we iterate over every edge in $A(u)$ in increasing order
of the {\em stored goodness value}, $\curgoodness(u,v)$ 
until we either find an edge whose approximate goodness $\apxgoodness(u,v)$
is smaller than $(1+\epsilon)$, or we exhaust all edges with
$\curgoodness(u,v)$ in the range $[1, (1+\epsilon)/(1+\alpha)]$.

Each edge unsuccessfully considered in this range has
$\curgoodness(u,v) \leq (1+\epsilon)/(1+\alpha)$, but
$\apxgoodness(u,v) > 1+\epsilon$.
When analyzing the running time, this gap helps show
that an edge can be unsuccessfully checked only $O(\log n)$ times.

\subsection*{Correctness}

\begin{lemma}\label{lem:merge}
The Merge$(u,v)$ operation preserves \cref{inv:weight,inv:good}.
\end{lemma}
\begin{proof}
Suppose that the invariants are true before the merge operation, and suppose
$u$ merges into $v$ (i.e. $v$ remains active at the end of the merge).

{\em \cref{inv:weight}.} The only edges whose weights are affected
by the merge are those incident to $v$. Let $(v,x)$ be such an edge.
If $(u,x)$ was in $G$ previously, then the weight of the $(v,x)$ edge will
be exact. Otherwise, \cref{inv:weight} previously held for $(v,x)$. It is clear
that \cref{inv:weight} continues to hold if the cluster size of $v$ has not yet
grown by a factor of $(1+\alpha)$; on the other hand if it grows by more than this
factor, $v$ broadcasts, and the $(v,x)$ weight will be exact.

{\em \cref{inv:good}.} Consider an arbitrary $(x, y)$ edge. When the edge was last assigned, say to $x$,
$\curgoodness(x, y) = \apxgoodness(x, y) \leq (1+\alpha)\exactgoodness(x, y)$.
The invariant can be violated if
$\exactgoodness(x, y)$ decreases, which can only occur by
$\max(\wmax(x), \wmax(y))$ decreasing.
The algorithm will update the goodness of $(x,y)$ once $\max(\apxwmax(x), \apxwmax(y))$ 
decreases by a $(1+\alpha)$ factor.
Since $\max(\apxwmax(x), \apxwmax(y)) \leq (1+\alpha)\max(\wmax(x), \wmax(y))$,
after $\max(\wmax(x), \wmax(y))$ decreases by a $(1+\alpha)$ factor,
$\max(\apxwmax(x), \apxwmax(y))$ must also decrease by a $(1+\alpha)$ factor.
At this point, we perform a broadcast and update the goodness of this edge to
$\apxgoodness(x, y)$. The ratio of $\curgoodness(x,y)/\exactgoodness(x, y)$ is 
always upper bounded by $(1+\alpha)^2$, and is largest
when initially $\exactgoodness(x, y) = \curgoodness(x, y)/(1+\alpha)$ and $\exactgoodness(u,v)$
decreases by a $(1+\alpha)$ factor.
\end{proof}

\begin{lemma}\label{lem:goodnessrange}
Once the algorithm terminates, no edges with goodness ($\exactgoodness$)
in the range $[1, T]$ with $T = (1+\epsilon)/(1+\alpha)^3$ remain.
\end{lemma}
\begin{proof}
Suppose otherwise. Let this edge be $(u,v)$ and its goodness value be $g_{uv}$.
Suppose $(u,v)$ is assigned to $u$.
By \cref{inv:good}, $\curgoodness(u,v) \leq \exactgoodness(u,v) (1+\alpha)^2$.
Since $g_{uv} \leq T = (1+\epsilon)/(1+\alpha)^3$, by
\cref{inv:good}, $\curgoodness(u,v) \leq g_{uv}(1+\alpha) \leq (1+\epsilon)/(1+\alpha)$.
However, this means that the algorithm incorrectly missed an
edge in the range $[0, (1+\epsilon)/(1+\alpha)]$ for node $u$,
contradicting the definition of BestAssignedNeighbor$(u)$.
\end{proof}

\subsection*{Running Time Analysis}
We will show that $O(\log n)$ total broadcasts and reassignment operations
are performed for any vertex over the course of the algorithm.

Let $a_1, \ldots , a_k$ be the sizes of the cluster containing
a vertex $a$. We have $a_1 = 1$ and $a_k \leq n$.
When a cluster of size $p$ merges with a cluster of size $q$ ($p \geq q$), 
the $\apxwmax(p)$ value can increase from $w$ to at most $w((1+\epsilon)q + p) / (p+q)$,
which is a multiplicative increase of $((1+\epsilon)q + p) / (p+q)$. 
The increase is maximized when $p = q$, for an
increase of at most $(2+\epsilon)/2$.
We would like to show that the overall increase in $\apxwmax$ over
any sequence of merges cannot be too large.
We will use the following lemma for the analysis:

\begin{lemma}
Let $a_1 < a_2 < \ldots < a_k$.
Then, $\prod_{i=1}^{k-1} (1+\frac{t (a_{i+1} - a_i)}{a_{i+1}}) = O(a_k^t)$.
\end{lemma}

\begin{proof}
Since $1+x \leq e^x$, for any $x_1, \ldots, x_k$ we have
\begin{equation*}\label{eq:1}
 \prod_{i=1}^k (1+x_i) \leq e^{\sum_{i=1}^k x_i}.
\end{equation*}
Thus, to complete the proof it suffices to show that
\begin{equation}
\sum_{i=1}^{k-1} \frac{t(a_{i+1} - a_i)}{a_{i+1}} \leq t \ln a_k + O(1),
\end{equation} 
and apply Equation~\ref{eq:1}. Observe that
\begin{equation*}
\frac{a_{i+1} - a_i}{a_{i+1}} = \sum_{j=0}^{a_{i+1} - a_i-1} \frac{1}{a_{i+1}} \leq \sum_{j=0}^{a_{i+1} - a_i-1} \frac{1}{a_{i+1} - j} = \sum_{j=a_i+1}^{a_{i+1}} \frac1j,
\end{equation*}
which implies
\begin{equation*}
 \sum_{i=1}^{k-1} \frac{t(a_{i+1} - a_i)}{a_{i+1}} \leq t \sum_{i=1}^{k-1} \sum_{j=a_i+1}^{a_{i+1}} \frac1j = t \sum_{j=1}^{a_k} \frac1j = t \ln a_k + O(1).
\end{equation*}
\end{proof}

By setting $t = \epsilon$, and $a_{i+1} = p + q$ and $a_i = p$,
$1 + \frac{t(a_{i+1} - a_i)}{a_{i+1}} = ((1+\epsilon)q + p)/(p + q)$.
Therefore, if $a_1 < a_2 < \ldots < a_k$ are the cluster sizes of a vertex $v$,
then, $\prod_{i=1}^{k-1} (1+\frac{\epsilon (a_{i+1} - a_i)}{a_{i+1}}) = O(a_k^\epsilon) = O(n^{\epsilon})$.
Therefore, the total increase in $\apxwmax(v)$ over the
course of the algorithm is at most $O(n^{\epsilon})$
Assuming a polynomially bounded aspect ratio, i.e., letting 
$\mathcal{W}_{\max}$ be the largest edge weight initially
and $\mathcal{W}_{\min}$ be the smallest edge weight initially,
we have that $\mathcal{W}_{\max} / \mathcal{W}_{\min} \in O(\mathsf{poly}(n))$.
Therefore the overall range for the best values of a
vertex over a sequence of merges is also $O(\mathsf{poly}(n))$.
For any constant $\alpha > 0$, the total number of broadcasts
performed by the vertex over its lifetime is $O(\log n)$. This is summarized
in the following lemma:

\begin{lemma}\label{lem:numbroadcast}
For any constants $\epsilon > 0, \alpha > 0$, and a graph with polynomially
bounded aspect ratio, the total number of broadcast operations performed by
any vertex over the sequence of the algorithm is $O(\log n)$.
\end{lemma}

Next, we show that the total number of times a given edge $(u,v)$
in unsuccessfully checked is low. We use the same constants and setting
(bounded aspect ratio):

\begin{lemma}
Each edge $(u,v)$ is checked
unsuccessfully over all calls to BestAssignedNeighbor$(u)$ 
at most $O(\log n)$ times.
\end{lemma}
\begin{proof}
Recall that a check of an edge $(u,v)$ is unsuccessful if $\curgoodness(u,v) \leq (1+\epsilon)/(1+\alpha)$
but $\apxgoodness(u,v) > (1+\epsilon)$.
Since $\curgoodness(u,v) = \apxgoodness(u,v)$ when this edge
was assigned to $u$, in the meantime, $\apxgoodness(u,v)$ grew
by at least a $(1+\alpha)$ factor. 
This could occur by $U = \apxwmax(u)$ growing
by a $(1+\alpha)$ factor, or $D = \min(\minmerge(u), \minmerge(v), \weight(u,v))$
decreasing by a $(1+\alpha)$ factor, or some combination of the two.
In either case, either $U$ must grow or $D$ must shrink by at least a $\sqrt{1+\alpha}$ factor i.e., a constant $> 1$.
Using the same idea as the proof of Lemma~\ref{lem:numbroadcast},
since $G$ has bounded aspect ratio, the total amount $U$ can grow is $\log_{1+\alpha} n^{1+\epsilon} = O(\log n)$, and so the number of unsuccessful checks to due $U$ is $O(\log n)$.
Similarly, since the three terms in $D$ can be at most $\mathcal{W}_{\max}$,
and can become no smaller than $\mathcal{W}_{\min} / n^2$, the overall
number of unsuccessful checks due to $D$ is also $O(\log n)$, completing the proof.
\end{proof}

\begin{theorem}
The overall running time of the algorithm is $O((m + n) \log^2 n)$.
\end{theorem}
\begin{proof}
Each broadcast requires updating the values of all neighbors of a node
$u$ in a heap at a cost of $O(\log n)$ per edge. Since each node requires
at most $O(\log n)$ broadcasts in total, the overall time is $O(m \log^2 n)$.
Similarly, the cost over all calls to BestAssignedNeighbor is $O(m \log^2 n)$
since each unsuccessful edge must be reinserted into a heap.
Lastly, the extract operations over the min heap indexed on the nodes costs
$O(n \log n)$ in total. Although the queue may be updated more times than this due to
cases where all edges in a call to BestAssignedNeighbor are unsuccessful, but
this cost can be charged to one of the unsuccessful edges.
\end{proof}

%Useful facts:
%\begin{itemize}
%\item If we use a heap-based approach indexed on all of the edges, we can only guarantee for each edge that it's a $(1+\alpha)^2$ approximation.
%
%\item If we use a heap-based approach indexed on all of the nodes then we can guarantee a $(1+\alpha)$ approximation since we only pay for approximation on the other endpoint's cluster size.
%\end{itemize}

\end{document}